\ifdefined\OPRE

\documentclass[opre,nonblindrev]{informs3}
\DoubleSpacedXI 

\ECRepeatTheorems

\EquationsNumberedThrough    

\MANUSCRIPTNO{OPRE-2020-07-507} 

\else
\documentclass[12pt]{article}

\newcommand{\ABSTRACT}[1]{
\begin{abstract}
#1
\end{abstract}
}

\newcommand{\ACKNOWLEDGMENT[1]}{
\section*{Acknowledgments}
#1
}

\usepackage[a4paper,margin=1in]{geometry}

\fi

\usepackage{breakcites}
\usepackage{amssymb,latexsym}
\usepackage{amsmath,graphicx}

\usepackage{graphics}
\usepackage{graphicx}
\usepackage{color}
\usepackage{euscript}

\usepackage{gensymb} 
\usepackage{csvsimple}  

\usepackage{mdframed}

\usepackage{hyperref}
\hypersetup{colorlinks=true,linkcolor=blue,citecolor=blue}

\numberwithin{equation}{section}

\newcommand{\range}[2]{\in\{#1,\dots,#2\}}

\newcommand{\dcg}[1]{\ensuremath{\overrightarrow{\EuScript{CG}}_{\mathbf{#1}}}}  
\newcommand{\ucg}[1]{\ensuremath{{\EuScript{CG}}_{\mathbf{#1}}}}  
\newcommand{\shar}[1]{\ensuremath{{\#\mathrm{shar}}_{\mathbf{#1}}}} 
\newcommand{\D}[1]{\ensuremath{{\mathrm{D}}_{\mathbf{#1}}}} 

\newcommand{\G}[1]{\ensuremath{{\mathcal{G}}_{\mathbf{#1}}}} 

\definecolor{ForestGreen}{rgb}{.13,.54,.13}
\definecolor{darkForestGreen}{rgb}{.05,.30,.05}
\definecolor{BrickRed}{rgb}{.80,.26,.33}
\definecolor{darkBrickRed}{rgb}{.70,.13,.16}

\ifdefined\DRAFT
\definecolor{ForestGreen}{rgb}{.13,.54,.13}
\newcommand{\erel}[1]{\textcolor{ForestGreen}{(\textbf{Erel:} #1)}}

\newcommand{\fed}[1]{\textcolor{red}{(\textbf{Fedor:} #1)}}
\else
\newcommand{\erel}[1]{}

\newcommand{\fed}[1]{}
\fi

\usepackage{natbib}
 \bibpunct[, ]{(}{)}{,}{a}{}{,}%
 \def\BIBand{and}%

  \let\theoremstyle\relax \usepackage{amsthm} 

\newtheorem{theorem}{Theorem}[section]
\newtheorem{lemma}[theorem]{Lemma}
\newtheorem{proposition}[theorem]{Proposition}
\newtheorem{corollary}[theorem]{Corollary}
\theoremstyle{definition}
\newtheorem{definition}[theorem]{Definition}
\newtheorem{remark}[theorem]{Remark}

\begin{document}

\ifdefined\OPRE


\RUNAUTHOR{Sandomirskiy and Segal-Halevi}

\RUNTITLE{Efficient Fair Division with Minimal Sharing}

\TITLE{Efficient Fair Division with Minimal Sharing}

\ARTICLEAUTHORS{%
\AUTHOR{Fedor Sandomirskiy}
\AFF{California Institute of Technology (Caltech), Pasadena, California, United States of America; 
Higher School of Economics, St.Petersburg, Russia;
\EMAIL{fsandomi@caltech.edu}} 
\AUTHOR{Erel Segal-Halevi}
\AFF{Ariel University, Ariel, Israel; \EMAIL{erelsgl@gmail.com}}
} 


\else

\title{Efficient Fair Division with Minimal Sharing}
\author{Fedor Sandomirskiy\footnote{Technion, IE\&M, Haifa, Israel; 
Higher School of Economics, St.Petersburg, Russia}
~and
Erel Segal-Halevi\footnote{
Ariel University, Ariel 40700, Israel
}
}
\maketitle

\fi

\ABSTRACT{
A collection of objects, some of which are good and some are bad, is to be divided fairly among agents with different tastes, modeled by additive utility functions.
If the objects cannot be shared, so that each of them must be entirely allocated to a single agent, then a fair division may not exist.
What is the smallest number of objects that must be shared between two or more agents in order to attain a fair and efficient division?

In this paper, fairness is understood as proportionality or envy-freeness, and efficiency, as fractional Pareto-optimality.
We show that, for a generic instance of the problem (all instances except a zero-measure set of degenerate problems), a \emph{fair fractionally Pareto-optimal division with the smallest possible number of shared objects can be found in polynomial time}, assuming that the number of agents is fixed. The problem becomes computationally hard for degenerate instances, where  agents' valuations are aligned for many objects.
}

\ifdefined\OPRE
\maketitle
\fi

\newpage

\section{Introduction.}
\label{sec:intro}

How to divide an inheritance among heirs or common assets among former partners in a fair way? This question becomes especially difficult if the division is to be made solely based on participants' individual preferences, while monetary compensations are excluded as it is often the case when the divided objects have emotional value for the participants.	
	Such questions are studied by fair division, a research area at the interface of economics and computer science; see a recent survey by \cite{moulin2019fair}.   What does it mean that objects are allocated fairly?  The two dominating notions of fairness in economic literature are \emph{envy-freeness}, which forbids the situation where some agent prefers the allocation received by somebody else to his/her own, and \emph{proportionality} requiring each agent to be as happy with  their own allocation as he/she would be in the equal division.
	
In economics, objects  are traditionally assumed to be some divisible resources. Under this assumption, fair allocations exist (e.g., the equal division), and, in a variety of settings, fairness can be combined with economic efficiency.
However, if the objects are indivisible, it may be impossible to allocate them fairly --- consider a single valuable  object and two people. 
Computer science literature focuses on indivisible objects and suggests looking for an \emph{approximately-fair} allocation in order to restore the existence. There are several definitions of approximate fairness, the most common of which is \emph{envy-freeness up to one object}.
An alternative approach to circumvent indivisibility, popular among economists, is to make indivisible objects ``divisible'' via randomization, ensuring that the division is fair ex-ante (in expectation).

While approximate or ex-ante fairness are reasonable when allocating low-value objects, such as seats in a course or in a school, they are not {suitable} for high-value objects, e.g., houses or precious jewels. Think of 
two siblings who have to divide three houses among them; it is unlikely that one of them would agree to receive a bundle that is envy-free up to one 
house, or a lottery that gives either one or two houses with equal probability.

Practical cases, however, fall in between the two extreme cases of divisible and indivisible objects. Usually, the objects can technically be shared among participants and, hence, can be treated as divisible, however, such sharing is  {undesirable}. For example, a shop can be jointly owned by several partners, sharing the costs and revenues of operation. A house can be jointly owned by several people, who live in the  house alternately. However, sharing may be quite inconvenient due to the overhead in managing the shared property. Unfortunately, sharing is inevitable for achieving exact fairness. Therefore, minimizing the number of objects  that have to be shared in a fair allocation becomes in practice an important concern.

\paragraph*{\textbf{Sharing minimization.}}
We propose a new approach to the problem of fair division:
\begin{center}
\emph{Treat objects as divisible, make sharing-minimization the objective,\\ and consider fairness and economic efficiency as constraints.}
\end{center}
This approach is a compelling alternative to approximate fairness when the objects to be divided are highly valuable and sharing is technically possible (as in all examples above and many other real-life situations) but unwanted.

We consider problems where the divided objects  may contain both goods and bads, as in the practice of partnership dissolution where valuable assets are often divided together with liabilities. We assume that agents  have additive  utilities, so the problem can be represented by a \emph{valuation matrix}, recording the value of each object to each agent. Such a ``bidding language'' is rather restrictive and does not allow to express  complementarity between objects (e.g., a garage becomes more valuable together with a car); however, it is a standard choice both in theory and practice since such preferences are easy to formulate and report~\citep{Goldman2015Spliddit}.

We assume that the number of agents $n$ is a fixed small number, as is common in inheritance cases and divorces, so the problem size is determined by the number of objects,  denoted by $m$.

We focus on the two classic fairness notions: \emph{envy-freeness} (each agent weakly prefers his/her bundle to the bundle of any other agent) and \emph{proportionality} (each agent gets a bundle worth at least $1/n$ of the total value, where $n$ is the number of agents). Economic efficiency is captured by  \emph{fractional Pareto optimality}: no other allocation, even without any restriction on sharing, improves the well-being of some agent without harming others. Note that even if no objects are shared in a given allocation  (hence, it can be considered as an allocation of indivisible objects), the  requirement of {fractional Pareto optimality} compares it to all divisible allocations. The word ``fractional'' is meant to stress this fact and contrast it to the requirement of discrete Pareto optimality which is popular in the literature on indivisible objects and is discussed below.

To the best of our knowledge, the idea of sharing minimization has not appeared in the literature. However, there are  known worst-case and average-case bounds on the number of shared objects; the gap between them {shows that there is} a significant room for optimization.
 For goods or bads (but not a mixture), \citet{bogomolnaia2016dividing} (Lemma~1) show that an envy-free fractionally Pareto-optimal allocation with at most $n-1$ sharings always exists. This upper bound is tight: when there are $n-1$ identical goods, all of them must be shared. 
 \citet{Dickerson2014Computational} considered a random instance with large number $m$ of indivisible goods and demonstrated that the allocation having the highest utilitarian welfare (achieved by allocating each object entirely to an agent with highest utility) is envy-free with high probability. Consequently, for large $m$, a fractionally Pareto-optimal allocation with $0$ sharings is likely to exist. These theoretical results suggest that in problems with many objects there is room for optimizing the number of sharings. Our computational experiments 
 on the real data from the fair division platform \url{spliddit.org} demonstrate that the number of sharings can often be optimized even for small $n$:
 more than 
 50\% of the 2-agents instances, 
 80\% of the 3-agents instances, 
 and 95\% of the 4-agents instances, 
 admit a proportional and fractionally Pareto-optimal allocation with fewer sharings than prescribed by the worst-case upper bound of~$n-1$
 (see Figure \ref{fig:spliddit-sharing}).

\paragraph*{\textbf{The results.}}
We consider
the algorithmic problem of
finding a fair and fractionally Pareto-optimal allocation
that minimizes the number of sharings. 
We find that the computational hardness of the problem depends on the 
\emph{degree of degeneracy} of the valuation matrix.
The degree of degeneracy measures how 
close the agents' valuations are to being identical.
Formally, it is the minimal number $k$ such that for each pair of agents, at most $k+1$ objects have the same ratio of agents' values.
The degree of degeneracy ranges between $0$ (for non-degenerate valuations) and $m-1$ (for identical valuations).

We demonstrate the following dichotomy (Theorem~\ref{thm:nagents-different-easy}):
\begin{itemize}
\item 
	Minimizing the number of sharings is \emph{algorithmically tractable if the degree of degeneracy of the valuation matrix is at most logarithmic in $m$}, 
	in particular, for non-degenerate valuations,
	we present an algorithm with a runtime \emph{polynomial in $m$} when $n$ is fixed.
\item Minimizing the number of sharings is
	\emph{NP-hard for any fixed $n\geq 2$ if the degree of degeneracy is at least of the order of $m^\alpha$ for some $\alpha>0$}, in particular, for identical valuations.
\end{itemize}

Since the set of valuations with positive degree of degeneracy has zero measure, our algorithm runs in polynomial time for almost all instances. However, its  runtime is exponential in the worst case, which is unavoidable due to the NP-hardness result. Despite this theoretical hardness, computational experiments on real data demonstrate the practical relevance of the algorithm.

When neither $n$ nor $m$ are fixed, a complementary hardness result is obtained in a follow-up paper by \citet{misra2021fair}. 
They demonstrate  that, for instances with degeneracy $1$, finding an fPO and EF allocation with minimal sharing (or even deciding the existence of an allocation with $0$ sharings) is NP-hard.

\paragraph*{\textbf{{Our methods and surprising sources of computational hardness. }}}
Our results
confirm the common sense that \emph{computationally-hard instances of resource-allocation problems are rare} and, in practice, such problems can often be  solved efficiently.

However, the fact that \emph{computationally-hard instances are those in which agents have identical valuations} 
is quite surprising.   
In many previous papers on fair division (e.g., \citet{oh2019fairly,plaut2018almost}), computational hardness results are presented with a qualifier saying that the problem is hard ``\emph{even} when the valuations are identical''; our results show that the ``even'' is unwarranted because {in our model} hard instances are exactly those with identical valuations.

Another observation that may seem surprising is that \emph{finding a fair and fractionally Pareto optimal allocation with minimal sharing is computationally easier than just fair with minimal sharing (without Pareto-optimality)}. The underlying reason is that, for non-degenerate problems, fractional-Pareto-optimality is a strong condition that shrinks the search space to a polynomial number of structures (see Proposition~\ref{prop:nagents-po} for a formal statement).
Without fractional Pareto-optimality, sharing-minimization becomes NP-hard, see Remark~\ref{rem_without_fPO_hard}. 
Sharing minimization without Pareto optimality is discussed in a follow-up paper \citep{segal2019fair}.

The polynomial ``size'' of the Pareto frontier for non-degenerate problems is the key observation, which allows us to conduct the exhaustive search over fPO  allocations. A similar insight underlies the algorithms of~\cite{devanur_kannan2008market} and~\cite{branzei2019chores} for computing equilibrium allocations of Fisher markets. We demonstrate the universal power of these ideas by showing the first application beyond equilibria of exchange economies. A follow-up paper \citep{aziz2020polynomial} builds on our results and demonstrates that the approach can also be used to design polynomial-time algorithms for approximately-fair fPO allocations of indivisible items.
The dynamic programming approach that we use to enumerate the Pareto frontier is more intuitive than that of~\cite{branzei2019chores} and does not rely on the ``black-box'' of cell-enumeration technique used by~\cite{devanur_kannan2008market}.
Another related observation is known within the framework of smoothed analysis of NP-hard problems: the Pareto frontier for a knapsack problem and its extensions becomes polynomially-sized if the instance is randomly perturbed (see \citet{Moitra_2011} for a survey).

Note the important \emph{contrast between fractional and discrete Pareto-optimality}, which is the dominant notion of economic efficiency in the literature on indivisible objects (an allocation is discrete Pareto-optimal if it is not dominated by any allocation with zero sharings). For discrete Pareto-optimality, even basic questions are computationally hard: deciding whether a given allocation is discrete-PO is co-NP complete, and deciding whether there exists an envy-free discrete-PO allocation is $\Sigma_2^p$-complete
\citep{deKeijzer2009Complexity}.
In contrast, deciding whether a given allocation is fractional-PO is polynomial in $m$ and $n$ (see Lemma~\ref{lem:po-check}), and deciding whether there exists an envy-free and fractionally PO allocation with no sharings is, for almost all instances, polynomial in $m$ for fixed $n$ (Theorem~\ref{thm:nagents-different-easy}).
These observations suggest that \emph{fractional-Pareto-optimality is a compelling concept of economic efficiency even for indivisible objects};
recent results by \cite{barman2018proximity,barman2018finding} support this conclusion.

Checking fractional Pareto-optimality efficiently  is a critical building block of our approach. An algorithm provided by Lemma~\ref{lem:po-check} runs in  strongly-polynomial time
	even if neither the number of agents $n$ nor the number of objects $m$ are bounded; this algorithm is  used to enumerate the Pareto frontier in strongly-polynomial time for fixed $n$.
	The algorithm is based on the equivalent ``dual'' condition for fractional Pareto-optimality: absence of profitable cyclic trades (Lemma~\ref{lem:po-cycle}). This dual characterization is known in the case of goods~\citep{bogomolnaia2016dividing, barbanel2005geometry} and in the case of bads~\citep{branzei2019chores}, but not in the case of a mixture. Extension to the mixture is not straightforward and requires a rather tricky definition of a weighted consumption graph, see Subsection~\ref{sub:agent-object-graphs}. As far as we know, algorithmic implications of the characterization are new even in the case of pure goods or pure bads, e.g.,  Lemma~\ref{lem:po-n-1}. We note that a weakly polynomial algorithm for checking fractional Pareto optimality can be deduced from the classic connection between fractional Pareto optimality and weighted welfare maximization; see Lemma~\ref{lem:po-weights} and the discussion around.

\paragraph*{\textbf{Fisher markets and competitive allocations.}} 
Fair division with divisible goods is tightly connected to the literature on competitive equilibria of exchange economies known as Fisher markets; see literature review in Section~\ref{sec:related}.  In such markets, agents are endowed with budgets of some artificial currency with no intrinsic value, and the goal is to find prices such the market clears when each agent purchases the most preferred among affordable bundles. The corresponding allocations are called competitive equilibria (CE).
\cite{Varian1974Equity} suggested using CE with equal budgets  as a fair division rule; it is denoted by CEEI since budgets are sometimes referred to as incomes. 
For general convex preferences, CEEI exists, it is envy-free and fractionally Pareto optimal; under the additional assumption of homogeneity, it maximizes the product of agents' utilities (the Eisenberg-Gale convex optimization problem). Both properties, convexity and homogeneity, are satisfied by additive valuations considered in our paper.
CEEI was extended to bads and a mixture of goods and bads by \cite{Bogomolnaia2017Competitive}.

 Any alternative approach to fair division of divisible objects is natural to compare with the  benchmark of CEEI. This benchmark is usually hard to  beat, yet our simulations on the data from  \url{spliddit.org} demonstrate that  the sharing-minimization approach  improves upon CEEI. 
{Namely, we find a proportional fPO allocation with fewer sharings in more than 36\% of the instances and envy-free fPO allocation, in more than 23\%  of them
 (see Table \ref{tab:compare-sharings}).}
This should not be surprising since  sharing minimization finds the best allocation among all fair fractionally Pareto optimal allocations, while CEEI picks a particular point in this set.
 
 By the First and Second Welfare Theorems, the set of fractionally Pareto optimal allocations and the set  of CE with arbitrary budgets coincide. Hence, sharing minimization can be reinterpreted as the optimization in the space of budgets with the number of sharings as the objective and fairness (either envy-freeness or proportionality) as a constraint. Unfortunately, nothing is known about the structure of the set of budgets such that  CE is fair except that it contains equal budgets. This obstacle precludes using the connection to CE to minimize sharing algorithmically. 
 
 Since the sets of CE and  fractionally Pareto optimal allocations coincide, all the structural results for the latter, such as Lemma~\ref{lem:po-cycle} and \ref{lem:po-n-1}, translate to the former and vice versa. In particular, undirected consumption graphs of fractionally Pareto optimal allocations correspond to the so-called maximal bang per buck (MBB) graphs for the Fisher market. Consequently, the algorithm enumerating all consumption graphs (Proposition~\ref{prop:nagents-po}) can be used to find all MBB graphs. A similar idea for enumeration of MBB graphs was employed by \cite{branzei2019chores} to compute CE for bads, where it no longer solves a convex optimization problem. 
 
 For goods, it is known that the consumption graph of a CE is acyclic or can be made acyclic via cyclic trades leaving all agents indifferent \citep{orlin2010improved,barman2018proximity}. For goods, combining this result with the Second Welfare Theorem, one can deduce the non-algorithmic part of Lemma~\ref{lem:po-n-1} and the ``only if'' part of Lemma~\ref{lem:po-cycle}. However, there is a common obstacle for obtaining the algorithmic part of Lemma~\ref{lem:po-n-1}, the ``if'' part Lemma~\ref{lem:po-cycle},  as well as  Lemma~\ref{lem:po-check} (algorithmic corollary of Lemma~\ref{lem:po-cycle})
 using the market approach.	Cyclic trades considered in the market literature are determined by equilibrium prices \citep[proof of Claim 2.2.]{barman2018proximity}, but neither prices nor budgets are given in the lemmas mentioned above. For mixed problems, the connection to the Fisher market becomes even less helpful since it is not known  if the Second Welfare Theorem holds in this case, and the First Welfare Theorem may be violated depending on nuances in the definition of CEEI \citep{Bogomolnaia2017Competitive}. For this reason, our proofs do not use  Fisher market techniques even though some claims might be proved via this alternative approach.

 \paragraph*{\textbf{Goods, bads, and mixed problems.}} Starting from the paper by \citet{Bogomolnaia2017Competitive}, it has been understood that  problems with bads (or a mixture of goods and bads) are structurally different from those with goods; see  \cite{moulin2019fair} for a survey. For example, there are impossibility results specific for  bads, CEEI for bads becomes multi-valued and no longer solves a convex optimization problem,  approximate-fairness guarantees for indivisible objects often differ in the case of goods and bads.
 
 With all this evidence contrasting fair division of goods with that of bads, our results provide an exception. Although the presence of both goods and bads complicates the constructions, it does not lead to conceptual obstacles and does not alter the results.

\paragraph*{\textbf{Structure of the paper.}}
In {\bf Section}~\ref{sec:model}, we introduce the notation and describe useful tools such as characterizations of fractional Pareto-optimality and worst-case bounds on the number of sharings. While most of these results are known for goods, extension to mixed problems is not straightforward.
Our results about sharing-minimization  are contained in {\bf Section}~\ref{sec:po-fair-algorithm}.  \textbf{Section~\ref{sec:simulations}} describes computational experiments on  real data from \url{spliddit.org}.
Related work is surveyed in \textbf{Section~\ref{sec:related}}. \textbf{Appendices~\ref{sec:po-cycle}} and~\textbf{\ref{sec:po-n-1}} are devoted to omitted proofs.

\section{Preliminaries.}
\label{sec:model}
\subsection{Agents, Objects and Allocations.}
There is a set $[n] = \{1,\ldots,n\}$ of $n$ agents and a set $[m]=\{1,\ldots,m\}$ of $m$ {divisible} objects.
A \emph{bundle} $\mathbf{x}$ of objects is a vector $(x_o)_{o\in [m]}\in[0,1]^m$, where the component $x_o$ represents the portion of $o$ in the bundle  (the total amount of each object is normalized to $1$).

Each agent $i\in[n]$ has an \emph{additive} utility function over bundles: $
u_i(\mathbf{x}) =
\sum_{o\in [m]}  v_{i,o}\cdot x_{o}
$. Here $v_{i,o}\in \mathbb{R}$ is agent $i$'s value of receiving the whole object $o\in[m]$; the matrix $\mathbf{v}=(v_{i,o})_{i\in [n],o\in[m]}$ is called the \emph{valuation matrix}; it encodes the information about agents' preferences and is used below as the input of fair division algorithms.

We make no assumptions on valuation matrix $\mathbf{v}$ and allow values of mixed signs: for example, the same object $o$ can bring positive value to some agents and negative to others. We say that an object $o$ is:
\begin{itemize}
\item a \emph{bad} if $v_{i,o}<0$ for all $i\in[n]$.
\item \emph{neutral} 
if $v_{i,o}=0$ for at least one $i\in [n]$ and $v_{j,o}\leq 0$ for all $j\in[n]$;
\item a \emph{good} if $v_{i,o}>0$ for at least one $i\in [n]$;
\begin{itemize}
\item a \emph{pure good} if $v_{i,o}>0$ for all $i\in [n]$;
\end{itemize}
\end{itemize}
Note the asymmetry in this notation: an object is a bad if \emph{everyone} think that it is bad, but an object is a good if at least \emph{one} agent thinks so. In other words, an object is a good if it can be allocated in a way that gives the owner positive utility, while a bad harms any owner.

An \emph{allocation} $\mathbf{z}$ is a collection of bundles $(\mathbf{z_i})_{i\in [n]}$, one for each agent, with the condition that all the objects are fully allocated. An allocation can be identified with the matrix $\mathbf{z} := (z_{i,o})_{i\in[n],o\in[m]}$  such that all $z_{i,o}\geq 0$ and $\sum_{i\in [n]} z_{i,o} = 1$ for each $o\in[m]$.

The \emph{utility profile} of an allocation  $\mathbf{z}$ is the vector $\mathbf{u(z)} := (u_i(\mathbf{z}_i))_{i\in [n]}$.

\subsection{Fairness and efficiency concepts.} The two fundamental notions of fairness, taking preferences of agents into account, are \emph{envy-freeness} and a weaker concept of \emph{proportionality} (also known as \emph{equal split lower bound} or \emph{fair share guarantee}).

An allocation $\mathbf{z}=(z_i)_{i\in[n]}$ is called \emph{envy-free (EF)} if every agent weakly prefers his/her bundle to the bundles of others. Formally, for all $i,j\in[n]$: $u_i(\mathbf{z}_i) \geq u_i(\mathbf{z}_j)$.

An allocation $\mathbf{z}$ is \emph{proportional (PROP)} if each agent prefers his/her bundle to the equal division: $\forall i\in[n]$ 
$u_i(\mathbf{z}_i) \geq \frac{1}{n}\sum_{o\in[m]}v_{i,o}$. Every envy-free allocation is also proportional; with $n=2$ agents, envy-freeness and proportionality are equivalent.

The idea that the objects must be allocated in an efficient, non-improvable way is captured by Pareto-optimality. An allocation $\mathbf{z}$ is \emph{Pareto-dominated} by an allocation $\mathbf{y}$ if $\mathbf{y}$ gives at least the same utility to all agents and strictly more to at least one of them.

An allocation $\mathbf{z}$ is \emph{fractionally Pareto-optimal (fPO)}
if no feasible $\mathbf{y}$ dominates it.
We note that the literature on indivisible objects considers a weaker notion of economic efficiency:  $\mathbf{z}$ is \emph{discrete Pareto-optimal} if it is not dominated by any feasible $\mathbf{y}$  with $y_{i,o}\in\{0,1\}$ for every agent $i$ and object $o$. While fractional-Pareto-optimality has good algorithmic properties, its discrete version does not, see the discussion in~Section~\ref{sec:intro}.

We will also need the following extremely weak but easy-to-check efficiency notion: an  allocation $\mathbf{z}$ is  \emph{non-malicious} if each good $o$ is consumed by agents $i$ with $v_{i,o}>0$
and each neutral object $o$ by agents $i$ with $v_{i,o}=0$. Non-malicious allocations impose no restriction on the allocation of bads.
Every fPO allocation is clearly non-malicious.

\subsection{Agent-object Graphs and a Characterization of fPO.}
\label{sub:agent-object-graphs}
Our algorithms use several kinds of \emph{agent-object graphs} --- bipartite graphs in which the nodes on one side are the agents and the nodes on the other side are the objects:

In the \emph{(undirected) consumption-graph \ucg{z}} of an allocation $\mathbf{z}$, there is an edge between agent $i\in[n]$ and object $o\in[m]$ iff $z_{i,o} > 0$.
	
The \emph{weighted directed consumption-graph \dcg{z}} of  $\mathbf{z}$
is constructed as follows.
There is an edge $(i\to o)$ with weight $w_{i \to o}=|v_{i,o}|$ if one of the two condition holds:
	\begin{itemize}
		\item $z_{i,o}>0$ and $v_{i,o}\geq 0$ (agent $i$ consumes a non-zero amount of $o$ and thinks that $o$ is not a bad)
		 \item $z_{i,o}<1$ and $v_{i,o}<0$ (a positive fraction of $o$ is not consumed by $i$, who treats $o$ as a bad).
	\end{itemize}	
The opposite edge $(o\to i)$ with weight $w_{o \to i}=\frac{1}{|v_{i,o}|}$ is included in  \dcg{z} in one of the two cases:
	\begin{itemize}
	\item $z_{i,o}>0$ and $v_{i,o}<0$ (agent $i$ consumes a non-zero amount of $o$ and perceives $o$ as a bad)
	\item $z_{i,o}<1$ and $v_{i,o}>0$ (a positive fraction of $o$ is not consumed by $i$, who thinks that $o$ is a good).
\end{itemize}

Intuitively, \dcg{z}
 captures the structure of possible exchanges in which agents may engage.
 Outgoing edges represent those objects that an agent $i$ 
 can use as a ``currency'' to pay others: either goods $i$ owns or bads owned by somebody else (in this case $i$ pays $j$ who owns a bad $b$ by taking some portion of $b$). Similarly, the incoming edges represent those objects that the agent can accept as a payment. A transaction involving an object $o$ and agents $i$ and $j$ such that $i$'s utility weakly decreases and $j$'s utility strictly increases is possible  if there are edges $i\to o\to j$.

 Incoming edges are those objects that $i$ is ready to accept as a currency: either to receive a valuable good, or to diminish $i$'s own bad.
An example is shown in Figure \ref{fig:consumption-graph}.

\begin{figure}
\begin{center}
\includegraphics[width=6cm]{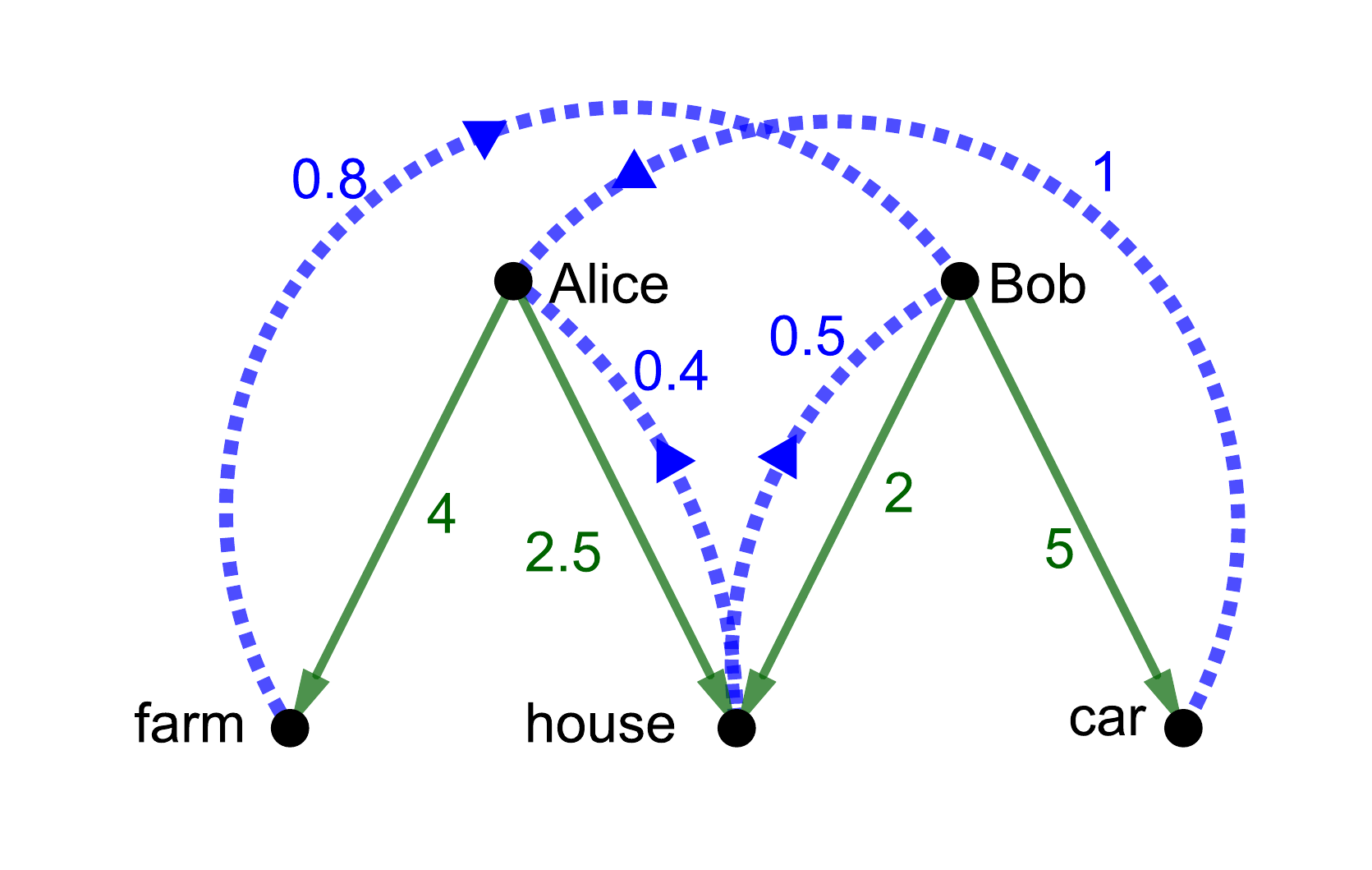}
~~~
\includegraphics[width=6cm]{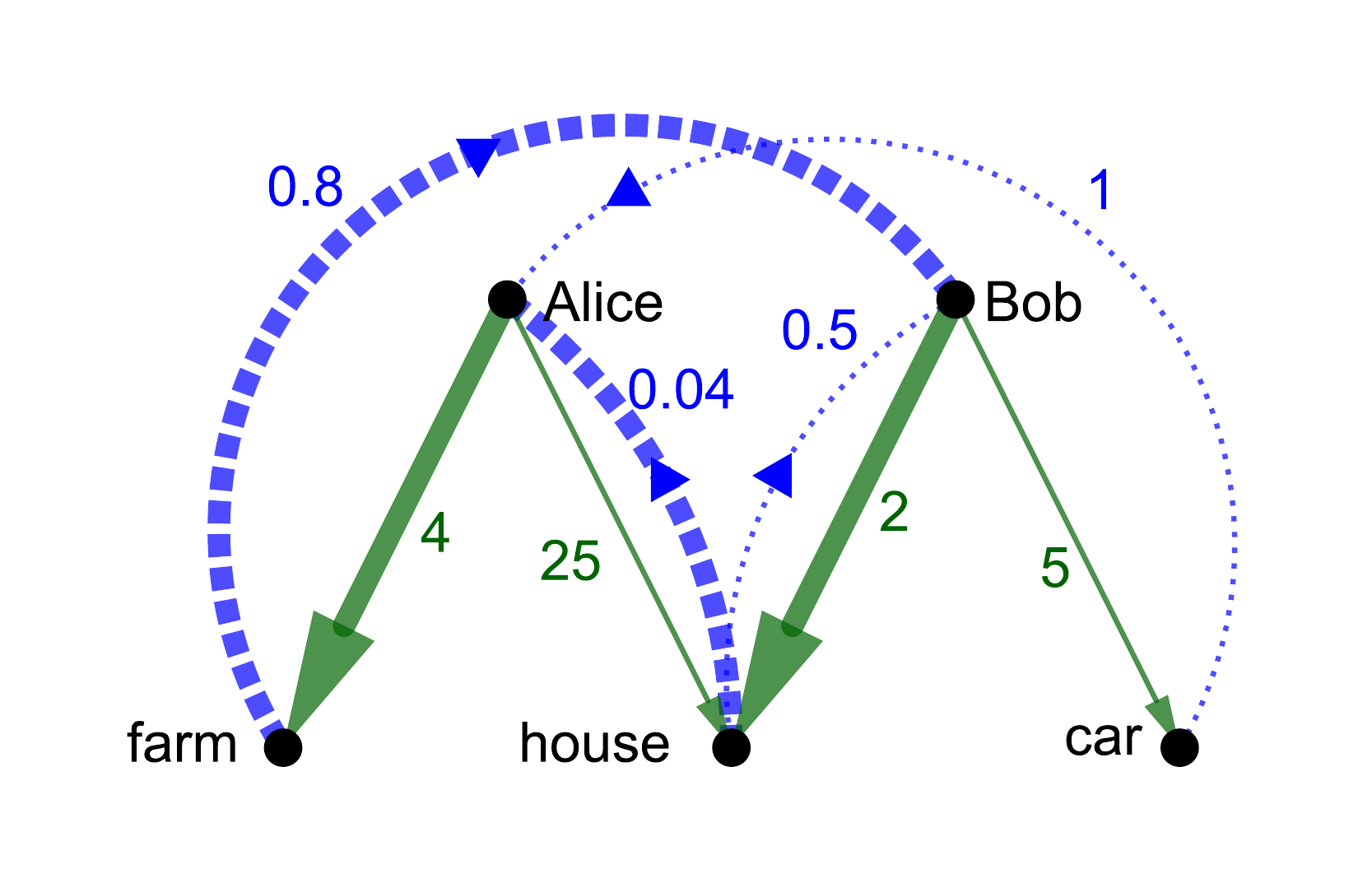}
\end{center}
\caption{\label{fig:consumption-graph} 
Some examples of weighted directed consumption graphs in an instance with pure goods. Alice gets the farm, Bob gets the car, and they share the house.
~~
\textbf{Left:} The farm, house and car are valued by Alice at 4, 2.5 and 1, and by Bob at 1.25, 2 and 5. The allocation is fPO.
~~
\textbf{Right:} The house is valued by Alice at 25 (all other valuations are the same). The allocation is not fPO;  a cycle $C$ with $\pi(C)<1$ is shown by bold arrows.}
\ifdefined\Description
\Description[examples of weighted directed consumption graphs]{Some examples of weighted directed consumption graphs in an instance with pure goods}
\fi
\end{figure}

Given $\mathbf{v}$,  one can reconstruct $\dcg{z}$ from $\ucg{z}$ and vice versa. Indeed, the condition $z_{i,o}<1$ from the definition of $\dcg{z}$ holds if and only if there is an agent $j\ne i$ with $z_{j,o}>0$, i.e., if $o$ is connected to some $j\ne i$ in $\ucg{z}$.

The \emph{product} of a directed path $P$ in \dcg{z}, denoted $\pi(P)$, is the product of weights of edges in $P$. In particular, the product of a cycle $C = (i_1 \to o_1 \to \allowbreak \ldots \to \allowbreak o_L \to i_{L+1} = i_1)$ is:

$$
\pi(C)
= \prod_{k=1}^{L}  \big(w_{i_k\to o_{k}}\cdot w_{o_k\to i_{k+1}}\big).
$$

The importance of this product is justified by the following lemma (proved in Appendix~\ref{sec:po-cycle}):
\begin{lemma}
\label{lem:po-cycle}
An allocation $\mathbf{z}$ is fractionally Pareto-optimal if-and-only-if 
it is non-malicious and
its directed consumption graph $\dcg{z}$ has no cycle $C$ with $\pi(C)<1$.
\end{lemma}

We see that the information about Pareto-optimality of an allocation is ``encoded'' in its consumption graph. An analog of Lemma~\ref{lem:po-cycle} is known for pure goods in a cake-cutting context~\citep{barbanel2005geometry} and was recently extended to problems with bads only by~\citet{branzei2019chores}.
Lemma~\ref{lem:po-cycle}
has a useful computational implication, which is also proved in Appendix \ref{sec:po-cycle}.
\begin{lemma}
\label{lem:po-check}	
It is possible to decide in time $O(nm(n+m))$ whether a given allocation $\mathbf{z}$ is fractionally Pareto-optimal.
\end{lemma}
Classic results in economic theory
\citep{negishi1960welfare,varian1976two}
represent the Pareto frontier of economies with convex sets of feasible utilities as the set of allocations $\mathbf{z}$ that maximize the weighted utilitarian welfare $\sum_{i\in [n]} \lambda_i u_i(\mathbf{z}_i)$ for some positive weights $\lambda_i$. This leads to another characterization of fPO allocations. 
The proof can be found in \citet{branzei2019chores} for the case of bads and in Appendix \ref{sec:po-cycle} for a mixture of goods and bads.
\begin{lemma}
	\label{lem:po-weights}	
	An allocation $\mathbf{z}$ is fractionally Pareto-optimal if and only if there is a vector of weights  $\lambda=(\lambda_i)_{i\in [n]}$ with $\lambda_i>0$, $i\in[n]$, such that
	\begin{equation}\label{eq_lambda}  
	z_{i,o}>0 \ \ \mbox{implies} \ \  \lambda_i v_{i,o}\geq \lambda_j v_{j,o}
	\end{equation}
	for all agents $i,j\in[n]$ and objects $o\in[m]$.
\end{lemma}
We note that Lemma~\ref{lem:po-weights} implies	an alternative algorithm for checking  fractional Pareto-optimality, which, in contrast to the strongly-polynomial algorithm of Lemma~\ref{lem:po-check} runs in weakly polynomial time. Indeed, fractional Pareto-optimality of $\mathbf{z}$ is equivalent to the existence of the vector $\lambda$ and the latter exists  if and only if the linear program formed  by  inequalities~\eqref{eq_lambda}  and the positivity condition is feasible (the positivity requirement can be replaced by $\lambda_i\geq 1$, $i\in[n]$, since the linear program is homogeneous).

Lemma~\ref{lem:po-weights} provides useful insights into the ``threshold'' structure of fractionally Pareto optimal allocations captured by the following necessary condition for fPO.
\begin{corollary}
	\label{corr:2agents-po}
	 For a fractionally Pareto-optimal allocation $\mathbf{z}$ and any pair of agents $i\ne j$, there is a threshold $t_{i,j}>0$ ($t_{i,j}=\frac{\lambda_j}{\lambda_i}$ from Lemma~\ref{lem:po-weights}) such that for any object $o$
	\begin{itemize}
		\item if $v_{i,o}\cdot v_{j,o}>0$ (i.e., both agents agree whether $o$ is a good or a bad), then
		\begin{itemize}
			\item  for $\frac{|v_{i,o}|}{|v_{j,o}|} > t_{i,j}$,  we have $z_{j,o}=0$ in the case of a good and  $z_{i,o}=0$ in the case of a bad
			\item for $\frac{|v_{i,o}|}{|v_{j,o}|} < t_{i,j}$, we have $z_{i,o}=0$ in the  case of a good and  $z_{j,o}=0$ in the case of a bad
		\end{itemize}
		\item if $v_{i,o}\cdot v_{j,o}<0$, then an agent with negative value cannot consume $o$   
		 \item if $v_{i,o}=0$ but $v_{j,o}\ne0$, then  $z_{j,o}=0$ for negative $v_{j,o}$ and $z_{i,o}=0$ for positive $v_{j,o}$.
	\end{itemize}
In particular, the only objects $o$ that can be shared between $i$ and $j$ are those with $\frac{v_{i,o}}{v_{j,o}} = t_{i,j}$ or with $v_{i,o}=v_{j,o}=0$.
\end{corollary}
A similar result underlies the Adjusted~Winner procedure introduced by~\citet{Brams1996Fair} for goods and extended to mixed problems in~\citep{aziz2018fair}.
The condition is necessary and sufficient for $n=2$ (see~\citep{bogomolnaia2016dividing} for either goods or bads). For $n\geq 3$, the condition is not sufficient; as one can see, it is equivalent to having no cycles of length $4$  with $\pi(C)<1$ in \dcg{z} (for $n=2$ any simple cycle has length at most $4$). However, it does not exclude longer cycles.

\subsection{Measures of Sharing and Worst-case Bounds.}\label{sec:po-fair-bounds}

If for some $i\in [n]$, $z_{i,o} = 1$, then the object $o$ is not shared --- it is fully allocated to agent $i$. 
Otherwise,  object $o$ is shared between two or more agents.
Throughout the paper, we consider two measures  quantifying the amount of sharing in a given allocation $\mathbf{z}.$

The simplest one is \emph{the number of shared objects} $\big|\left\{o\in[m]\, :\, z_{i,o}\in(0,1)\mbox{ for some }i\in[n]\right\}\big|.$
Alternatively, one can take into account the number of times each object is shared. This is captured by \emph{the number of sharings}

\begin{align*}
\shar{z}=\sum_{o\in[m]}\left(\big|\{i\in [n]:\, z_{i,o}>0\}\big|-1\right).
\end{align*}
For discrete allocations, {where each object is allocated entirely to one of the agents,} both measures are zero.  They differ, for example, if only one object $o$ is shared but each agent consumes a bit of $o$: the number of shared objects is then $1$ while the number of sharings is $n-1$.
Clearly, the number of shared objects is always at most the number of sharings.

When there are $n$ agents and $n-1$ identical pure goods, a fair allocation must give each agent a fraction $(n-1)/n$ of a good, for any reasonable definition of fairness. This requires sharing all $n-1$ goods, so $n-1$ is a worst-case lower bound on the number of shared objects and thus also for the number of sharings.
This lower bound can always be attained.

In the case of pure goods or bads but not a mixture, \citet{bogomolnaia2016dividing} showed that for any fractionally Pareto-optimal allocation $\mathbf{z}$, there exists  an equivalent one $\mathbf{z^*}$ (in both allocations, all agents receive the same utilities) with  $\shar{z^*}\leq n-1$. A similar result can be found in \citep{barman2018proximity} (see Claim 2.2) for the so-called competitive equilibrium allocations of pure goods.
The following lemma shows that $\mathbf{z^*}$ can be constructed efficiently. It is proved in Appendix \ref{sec:po-n-1} and covers arbitrary allocations $\mathbf{z}$ rather than just fractionally Pareto-optimal, and a mixture of goods and bads. 
\ifdefined\EC
\else
{In a follow-up work~\citep{aziz2020polynomial} this lemma was applied for computing almost-proportional fPO allocations of such a mixture.}
\fi
\begin{lemma}
\label{lem:po-n-1}
For any allocation $\mathbf{z}$, there is
a fractionally PO allocation $\mathbf{z^*}$ such that:
\begin{itemize}
\item (a) $\mathbf{z^*}$ either Pareto dominates $\mathbf{z}$ or gives every agent the same utility as $\mathbf{z}$.
\item (b) the undirected consumption graph $\ucg{z^*}$ is acyclic.
\item (c) $\mathbf{z^*}$ has at most $n-1$ sharings (hence at most $n-1$ shared objects).
\end{itemize}
The allocation $\mathbf{z^*}$ can be constructed in  strongly-polynomial time using $O(n^2 m^2 (n+m))$ operations. 
\end{lemma}

This lemma, combined with known algorithms for computing fair and fractionally PO allocations, yields the following corollary.
\begin{corollary}
\label{corollary:po-ef-n-1}
In any instance with $n$ agents, there exists a fractionally Pareto-optimal, envy-free (and thus proportional) division with at most $n-1$ sharings.
In some special cases, such an allocation can be found in strongly-polynomial time:
\begin{itemize}
\item 
If the fairness notion is proportionality only --- using $O(n^2 m^2 (n+m))$ operations;
\item If all objects are pure goods --- {using $O\big((n+m)^4\log(n+m) + n^2 m^2 (n+m)\big)$ operations;} 
\item If objects are mixed --- using 
$O\big(m^{n+2}\big)$ operations for fixed $n$ and $O\big(n^{m+2}\big)$ for fixed $m$.
\end{itemize}
\end{corollary}

\begin{proof}
The existence of envy-free fPO allocations for mixed problems was proved by~\citet{Bogomolnaia2017Competitive}. Such allocations can be obtained as
competitive equilibria of the associated Fisher market with equal incomes (CEEI); see the discussion in Section~\ref{sec:intro} ans~\ref{sec:related}. CEEI satisfies the property of Pareto-indifference: if $\mathbf{z}$ is a CEEI and $\mathbf{z}^*$ gives the same utilities to all agents, then $\mathbf{z}^*$ is also a CEEI. This allows us to apply Lemma~\ref{lem:po-n-1} and get an envy-free fPO allocation with $\shar{z}\leq n-1$.

For the first claim, consider the equal-split allocation $\mathbf{z}$ ($z_{i,o}=\frac{1}{n}$ for all $i,o$) and construct a fractionally Pareto-optimal dominating allocation $\mathbf{z^*}$ by Lemma~\ref{lem:po-n-1}. Pareto-improvements preserve proportionality, and thus $\mathbf{z^*}$ is proportional, fPO, and has at most $n-1$ {sharings}.

Algorithms for computing Fisher market equilibria are known 
for pure goods with run time   $O((n+m)^4\log(n+m))$ \citep{orlin2010improved} and for a mixture
with 
run time $O\big(m^{n+2}\big)$ for  fixed $n$ and
$O\big(n^{m+2}\big)$  for fixed $m$ \citep{garg2020computing}.
Application of Lemma~\ref{lem:po-n-1} to the computed CEEI yields the remaining claims.
\end{proof}


\section{Pareto-Optimal Fair Division: Minimizing the Sharing.}
\label{sec:po-fair-algorithm}
As we saw in Subsection~\ref{sec:po-fair-bounds},  having $n-1$ sharings is unavoidable in the worst case. However, the average-case behavior is much better than the worst-case and it is likely to find a fair and fractionally Pareto optimal allocation with fewer sharings (see~\citep{Dickerson2014Computational} and the discussion in Section~\ref{sec:intro}).
This raises the following computational problem:

\begin{quote}
\centering
\emph{For a given instance of a fair division problem,\\find a solution that minimizes the number of sharings.
}
\end{quote}

We will contrast between two extreme cases: agents with \emph{identical valuations} and agents with \emph{non-degenerate valuations}.
\begin{definition}
\label{def:2agents-deg}
A valuation matrix $\mathbf{v}$ 
is called \emph{degenerate}
if there exist two agents $i,j$ and two objects $o,p$ such that $v_{i,o}\cdot v_{j,p} = v_{i,p}\cdot v_{j,o}$ (or $\frac{v_{i,o}}{v_{j,o}}=\frac{v_{i,p}}{v_{j,p}}$ if denominators are non-zero). %
Otherwise, it is called \emph{non-degenerate}.
\end{definition}
Note that if $\mathbf{v}$ is selected randomly according to a continuous probability distribution, it is non-degenerate with probability $1$.%

\citet{bogomolnaia2016dividing} and \citet{branzei2019chores} use a stronger definition of degeneracy: the complete agent-object graph has no cycles $C$ with $\pi(C)=1$. Their condition implies that $\ucg{z}$ is acyclic for any fPO allocation $\mathbf{z}$ and that there is a bijection between Pareto-optimal utility profiles and fPO allocations. 
Our definition addresses only cycles of length $4$ and thus can be  easily  checked in $O(n^2\cdot m\log m)$ operations (see Subsection~\ref{sub:nagents}).
For $2$ agents the definitions coincide.

\subsection{Warm-up: Two Agents, Pure Goods.}
\label{sub:2agents}
For $n=2,$ the upper bound on the number of sharings from Subsection \ref{sec:po-fair-bounds} is $1$, so sharing-minimization boils down to finding a fair allocation with no sharings at all (if such an allocation exists). The following ``negative'' result is well-known (e.g., \citep{Lipton2004Approximately}); we present it to contrast with the ``positive'' theorem after it.
\begin{theorem}
\label{thm:2agents-identical-hard}
When there are $n=2$ agents with identical valuations over $m$ pure goods, it is NP-hard to decide whether there exists an allocation with no sharings that is EF ($=$PROP for $n=2$) and fractionally-PO.
\end{theorem}

\begin{proof}
For two agents with identical valuations, all allocations are fractionally-PO. 
Thus, a fair+fPO allocation exists if-and-only-if the set of goods can be partitioned into two subsets with the same sum of values. Hence, the problem is equivalent to the NP-complete problem \textsc{Partition}.
\end{proof}


The following theorem shows that,
under the requirement of fractional Pareto-optimality, 
the computational problem becomes easier when the valuations are different.
\begin{theorem}
\label{thm:2agents-different-easy}
For two agents with non-degenerate 
valuations over $m$ pure goods, it is possible to find in time $O(m\cdot \log(m))$ a division that is EF ($=$PROP for $n=2$) and fractionally-PO,
and subject to these requirements, minimizes the number of sharings.

If the goods are pre-ordered by the ratio $\frac{v_{1,o}}{v_{2,o}}$, the computation takes linear time~$O(m)$.
\end{theorem}
\begin{proof}
Order the goods in descending order of the ratio ${v_{1,o} / v_{2,o}}$, for $o\range{1}{m}$ (this takes $O(m\log(m))$ operations). By the assumption of non-degeneracy, no two ratios coincide.

By Corollary~\ref{corr:2agents-po}, any fractionally PO allocation $\mathbf{z}$ takes one of two forms: 
\begin{itemize}
\item ``$0$ sharings'': there is a good $o$ such that $\mathbf{z}$ gives all the prefix goods $1,\dots,o$ to agent $1$, and all suffix goods $o+1,\dots,m$ to agent $2$.
\item ``$1$ sharing'': there is a good $o$ which is split between the two agents, while all goods $1,\dots,o-1$ are consumed by agent $1$ and all remaining goods $o+1,\dots,m$ by agent $2$.
\end{itemize}
Therefore, we have $m+1$ allocation with $0$ sharings and each of them can be tested for fairness. 
~~~
If there are no fair allocations among them, then we look for a fair allocation among those with one sharing. For any fixed $o$, this leads to solving a system of two linear inequalities  with just one variable (the amount of $o$ consumed by agent $1$).
\end{proof}

\begin{remark}\label{rem_without_fPO_hard}
If  the requirement of fractional-PO is removed,  we cannot escape the hardness of Theorem~\ref{thm:2agents-identical-hard} even for non-degenerate valuations. This can be demonstrated by adding a small twist in the proof of Theorem~\ref{thm:2agents-identical-hard}: instead of an instance with identical valuations, a tiny perturbation of it is considered; this eliminates degeneracy but preserves reduction from \textsc{Partition}. 
\ifdefined\EC
\else
Hardness of sharing-minimization without fPO is further explored by \citet{segal2019fair}.
\fi
\end{remark}

\subsection{Main Results: n Agents, Mixed Valuations, Varying Degeneracy.}
\label{sub:nagents}

Now we come back to the full generality of mixed problems with an arbitrary number of agents.

In order to capture instances ``in between'' the two extremes of non-degenerate  and identical valuations, we  define \emph{the degree of degeneracy}  of a valuation matrix $\mathbf{v}$ as 

$$\D{v}=\max_{i, j\in [n], i\ne j}\max_{r>0} \big|\big\{o\in [m]\, : \ v_{i,o}=r\cdot v_{j,o} \big\}\big|-1.$$

Informally, $\D{v}+1$ is the maximal number of objects $o$ such that some agents $i\ne j$ have the same ratio $\frac{v_{i,o}}{v_{j,o}}$ for all of them. Degree of degeneracy can be easily computed in time $O(n^2\cdot m\log(m))$: for each pair of agents rearrange the ratios in a weakly-decreasing order and then find the longest interval of constancy. 

A valuation matrix $\mathbf{v}$ is non-degenerate if and only if $\D{v}=0$. In particular, $\D{v}$ equals zero with probability one for any continuous
probability measure on
$\mathbb{R}^{n\times m}$.
In the case of identical valuations, $\D{v}$ attains its maximal value, which is $m-1$.

\begin{remark}
\label{rem:sharing-upper-bound}
By Corollary \ref{corr:2agents-po}, 
in any fPO allocation, 
for each pair of agents, 
at most $\D{v}+1$ objects are shared.
Hence,  in \emph{any} fPO allocation, 
the number of sharings is at most 
$(\D{v}+1)\frac{n(n-1)}{2}$.
Contrast this with Lemma \ref{lem:po-n-1}: it says that, for any fPO utility profile, \emph{there exists} an fPO allocation with these utilities, in which the number of sharings is at most $n-1$.
\end{remark}

The next theorem is our main result: it shows how increasing \D{v} moves us gradually from 
the easiness illustrated (for two agents) by Theorem~\ref{thm:2agents-different-easy}  to the hardness of  Theorem~\ref{thm:2agents-identical-hard}.
\begin{theorem}
	\label{thm:nagents-different-easy}
	 Fix the number of agents $n\geq 2$. 

(a) Given an $n$-agent instance $\mathbf{v}$ with a mixture of goods and bads,
		an allocation $\mathbf{z}$  that minimizes the number of sharings
		 \shar{z}  subject to 
		fractional-Pareto-optimality and proportionality (or envy-freeness) can be computed using 
				$$O\left(3^{\frac{n(n-1)}{2}\cdot \D{v}} \cdot m^{\frac{n(n-1)}{2}+2}\right)$$
		operations.
		In particular, for  any fixed constant $C>0$, sharing-minimization can be performed in strongly-polynomial time for any instance $\mathbf{v}$ with $\D{v}\leq C\cdot \log(m)$.
		
(b) Fix arbitrary constants $C>0$ and $\alpha>0$. Checking the existence of a fractional-Pareto-optimal  proportional (or envy-free) allocation $\mathbf{z}$ with $\shar{z}=0$ is NP-hard for valuations $\mathbf{v}$ such that $\D{v}\geq C\cdot m^\alpha$.
\end{theorem}

\begin{proof}
\textbf{(a)} The algorithm has two phases. The first phase is to enumerate the set $\G{v}$ of all
 \emph{{fPO graphs}} ---  undirected consumption graphs  of fractionally PO allocations. This phase is the subject of Proposition~\ref{prop:nagents-po} below. 

The second phase is testing each  $G\in \G{v}$:
\begin{enumerate}
\item 
\label{item:upperbound}
Count the number of sharings in $G$. If it exceeds $n-1$, skip $G$ %
.
\item 
\label{item:variables}
For each shared object $o$, and for each agent $i$ connected to $o$, create a non-negative variable $z_{i,o}$ representing the fraction of $o$ allocated to $i$. The total number of such variables is at most $2(n-1)$ 
--- for each shared object, we have one variable for each agent connected to it. 
\item 
\label{item:lp}
Represent the required fairness condition (EF / proportionality) as a set of linear inequalities in these variables. Solve the resulting LP. 
\item 
Among those graphs $G$ where the {LP} has a solution, select the one with the smallest number of sharings and return the corresponding allocation. 
\end{enumerate}
Step \ref{item:upperbound} is justified by 
Lemma \ref{lem:po-n-1}: it ensures that we can restrict our attention to fPO allocations with at most $n-1$ sharings. Since all graphs of such allocations are checked, 
a fair fPO allocation with the minimal \shar{z} will be found.

Note that a solution to the linear program from Step~\ref{item:lp}, if exists, may not be unique. We stress that the algorithm can pick any solution to this LP, not necessary  the one that minimizes the number of sharings over the set of  solutions. Indeed, if for a graph  $G\in \G{v}$, the LP has a solution such that some variables $z_{i,o}$ are equal to zero, such a solution will be discovered when the algorithm faces the graph $G'\in \G{v}$, where the corresponding edges $(i,o)$ are eliminated.

For fixed $n$, the number of operations per fPO graph $G$ is $O(m)$, the time needed to ``read'' it.  Solving the LP takes constant  time since its size does not depend on $m$ --- it depends on $n$ only and $n$ is fixed: we have at most $2(n-1)$ variables, at most $3(n-1)$ feasibility constraints (at most $2(n-1)$ of non-negativity and at most $n-1$ of full allocation), $n$ fairness constraints for proportionality and $n(n-1)$ for EF. Thus, the run time of the second phase is $O(m\cdot |\G{v}|)$ and the first phase determines the overall complexity.\smallskip \\
\noindent \textbf{(b)} We  outline a reduction from  \textsc{Partition}. It is slightly more complicated than the reduction by \cite{Lipton2004Approximately} used in Theorem~\ref{thm:2agents-identical-hard} since we need to construct an instance that is not too degenerate and, in particular, we cannot rely on instances with identical preferences.

 We present the construction for $n=2$; the case $n>2$ can be covered by adding dummy agents.
Given an instance $a_1,a_2,\dots,a_p$ of \textsc{Partition}, pick a minimal $m=m(p)$ such that $C\cdot m^\alpha\geq p-1 $ %
. 
Define a fair division instance with $m$ pure goods of two types:
\begin{itemize}
\item $p$ ``big'' goods: for each $o\in[p]$, the good $o$ is equally valued by both agents: $v_{1,o}=v_{2,o}=a_o$.
\item $m-p$ ``small'' goods: 
there are $Q=\frac{m-p}{2}$ pairs (w.l.o.g., $m-p$ is even)  of goods $(q_k,\bar{q}_k)_{k\in[Q]}$ such that $v_{1,q_k}=v_{2,\bar{q}_k}=\frac{k+1}{4mk}$, $v_{1,\bar{q}_k}=v_{2,{q_k}}=\frac{1}{4mk}$.
\end{itemize}
Note that the value-ratios of the $p$ big goods are all equal to $1$, while the value-ratios of the small goods are all different ($k+1$ for $k=1,2,\ldots$). Hence the degeneracy degree of the instance is $p-1$.

The sum of each pair of small goods is less than $1/m$, so the sum of all small goods is less than $1/2$ for both agents, 
while the value of each big good is a positive integer. Therefore,
in any fair fractionally PO allocation, both agents consume some of the big goods. Thus, by Corollary~\ref{corr:2agents-po},  agent $1$ consumes all the $q_k$ goods (since their value-ratio is more than $1$) and agent $2$ all the $\bar{q}_k$ goods (since their value-ratio is less than~$1$).

Thus, a fair fPO allocation with $0$ sharings exists if and only if $a_1,\dots,a_p$ can be partitioned into two subsets of equal sum.

We have reduced \textsc{Partition} to an allocation problem with $\D{v}=p-1\leq C\cdot m^\alpha$. Since $m=m(p)$ is bounded by a polynomial in $p$, the length of binary representation of $\mathbf{v}$ is bounded by a polynomial of the size of \textsc{Partition} instance.
\end{proof}

The following proposition completes the proof of Theorem~\ref{thm:nagents-different-easy}:
it shows that the set of consumption graphs of all fPO allocations can be efficiently enumerated.
\begin{proposition}
	\label{prop:nagents-po}
	For every fixed number of agents $n\geq 2$, the set of all fPO graphs $\G{v}:=\left\{\ucg{z}\,:\ \mathbf{z} \mbox{ is fPO for } \mathbf{v} \right\}$ can be enumerated using $O\left(3^{\frac{n(n-1)}{2}\cdot \D{v}} \cdot m^{\frac{n(n-1)}{2}+2}\right)$ operations. In particular, for  $\mathbf{v}$ with logarithmic degeneracy ($\D{v}\leq C\cdot\log(m)$ as in Theorem~\ref{thm:nagents-different-easy}), the algorithm runs in strongly-polynomial time.
	
The total number of graphs in $\G{v}$ satisfies the upper bound\,
\begin{equation}\label{eq_number_PO_ucg}
	\big|\G{v}\big|\leq 3^{(1+\D{v})\frac{n(n-1)}{2}}\cdot m^{\frac{n(n-1)}{2}}.
	\end{equation}
\end{proposition}
Proposition~\ref{prop:nagents-po} is proved by the following two lemmas.  {
We enumerate all fPO graphs by iteratively adding agents. 
We start by enumerating fPO graphs for the first two agents (Lemma \ref{lem:enumerate-level2} and Figure \ref{fig:cg-tree-level2}).
Then we show that, given all fPO graphs for agents $1,\ldots,k$, we can efficiently enumerate all fPO graphs for agents $1,\ldots,k+1$
(Lemma \ref{lem:enumerate-levelk} and Figure \ref{fig:cg-tree-level3}).
}
\begin{figure}
\begin{center}
\includegraphics[width=0.8\textwidth,clip, trim = 2cm 1cm 2cm 1cm]{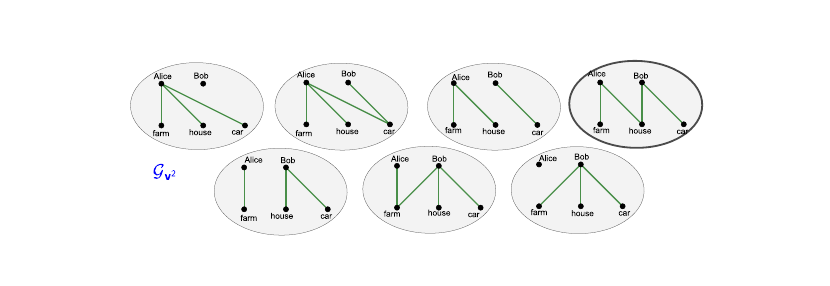}
\end{center}
\caption{
\label{fig:cg-tree-level2}
Enumerating the fPO consumption graphs of allocations between two agents (Lemma \ref{lem:enumerate-level2}), in a non-degenerate instance with $m=3$ pure goods, where the agents' valuations are as in Figure \ref{fig:consumption-graph}/Left.
There are $2 m + 1 = 7$ consumption-graphs of fPO allocations between Alice and Bob.
The emphasized graph at the top-right is expanded in Figure \ref{fig:cg-tree-level3} below.
}
\ifdefined\Description
\Description[Enumerating the fPO consumption graphs of allocations between two agents]{Enumerating the fPO consumption graphs of allocations between two agents}
\fi
\end{figure}

We note that there are several alternative approaches to enumerate fPO consumption graphs. \citet{branzei2019chores} recover a subset of fPO graphs by their $2$-agent projections in order to compute the so-called competitive allocations of bads. \citet{devanur_kannan2008market} use a complicated technique of cell-enumeration from computational algebraic geometry for a similar problem with goods. An alternative dynamic-programming approach for goods was outlined by~\citet{dw2019enumerate}. That algorithm  sequentially adds new goods and presumably runs in time $O(2^n m^2 |\G{v}|)$, 
however, the construction does not provide an a-priori polynomial upper bound on~$|\G{v}|$.

\begin{lemma}\label{lem:enumerate-level2}
For an instance $\mathbf{v}$ with $2$ agents and $m$ objects: (a) the total number of graphs in $\G{v}$ is at most $3m\cdot 3^{\D{v}}$.
~~(b)
If all the objects $o$ with $v_{1,o}\cdot v_{2,o}> 0$, are ordered by the ratio $\frac{|v_{1,o}|}{|v_{2,o}|}$, then  $\G{v}$ can be enumerated using  $O\left(m\cdot 3^{\D{v}}\right)$ operations.
\end{lemma}	
\begin{proof}
We generalize the construction used in Theorem~\ref{thm:2agents-different-easy}, from pure goods and non-degenerate valuations, to arbitrary objects and an arbitrary degeneracy level.
	
We use the following notation: 
\begin{itemize}
\item $A_{>0}=\{o\in[m] \ : \ v_{1,o}>0, \ v_{2,o}>0\}$ for the set of pure goods;
\item  $A_{<0}=\{o\in[m] \ : \ v_{1,o}<0,\ v_{2,o}<0\}$ for the set of bads,  
\item $A_{0}=\{o\in[m]\ : \  v_{1,o}=v_{2,o}=0\}$ for zero-valued objects, 
\item $A_{\pm}$ for all the remaining {impure goods and} neutral objects;
\item 
For a given positive number $t$, we define
$A_=(t) := 
\big\{ o\in A_{>0}\cup A_{<0}\ : \ \frac{|v_{1,o}|}{|v_{2,o}|}=t \big\}$.
\end{itemize}

(a) 
By non-maliciousness,
there is no flexibility in allocating objects from $A_{\pm}$: they are consumed by the agent with larger $v_{i,o}$ at any fPO allocation. In contrast, objects from $A_0$ can be allocated arbitrarily. Such zero objects contribute $|A_0|$ to $\D{v}$ and lead to $3^{|A_0|}$ allocation possibilities (each object is consumed either by agent $1$, or by agent $2$, or by both).
Note that if we were only interested in final allocations, then zero objects could be given to one of the agents arbitrarily. 
However, in Lemma \ref{lem:enumerate-level2} we count all possible consumption graphs, which are important for the later steps of the algorithm. Therefore, we consider all options for zero objects too.

The allocation of objects from $A_{>0}\cup A_{<0}$ is determined by the value-ratio threshold $t_{1,2}$ of Corollary~\ref{corr:2agents-po}. 
Consider the set $T=\left\{\frac{|v_{1,o}|}{|v_{2,o}|},\ o\in A_{>0}\cup A_{<0}\right\}$. To cover all the fPO allocations, it is enough to consider $|T|$ situations, when $t_{1,2}$ equals one of the elements of $T$. Then objects $o\in A_{>0}$ with $\frac{v_{1,o}}{v_{2,o}}>t_{1,2}$ are allocated to agent $1$ and with $\frac{v_{1,o}}{v_{2,o}}<t_{1,2}$ to agent $2$; symmetrically, for $o\in A_{<0}$, bads with $\frac{v_{1,o}}{v_{2,o}}>t_{1,2}$ go to agent $2$, while those with $\frac{v_{1,o}}{v_{2,o}}<t_{1,2}$ to agent $1$. The remaining objects $A_{=}(t_{1,2})$ are allocated arbitrarily between agents, resulting in $3^{|A_{=}(t_{1,2})|}$ possibilities.
All in all:

$$|\G{v}|\leq  3^{|A_0|}\cdot |T|\cdot 3^{\max_{t\in T} |A_=(t)|}.$$
Since $|T|\leq m$ and $|A_0|+|A_=(t)|\leq 1+\D{v}$ for any $t>0$,  we get the claimed upper bound. 
Note that the bound is not tight --- for $2$ agents and non-degenerate $\mathbf{v}$,  we can get $2m+1$ instead of $3m$ (see the proof of Theorem~\ref{thm:2agents-different-easy}).

(b) If $t$ and $t'$ are two consecutive elements of $T$, then passing from $t$ to $t'$ involves reallocation of objects from $A_{=}(t)$ and $A_{=}(t')$ only. This leads to an overall running time proportional to the total number of graphs $|\G{v}|$.%
\end{proof}

Given the valuation matrix $\mathbf{v}=(v_{i,o})_{i\in[n],o\in[m]}$ {and $k\in[n]$}, denote by $\mathbf{v}^{k}$ the valuation of the first $k$ agents: $\mathbf{v}^k=(v_{i,o})_{i\in[k],o\in[m]}$. 
The previous lemma
tells that $|\G{\mathbf{v}^{2}}|\leq 3^{\D{\mathbf{v}}+1}\cdot m$.
The next lemma relates $\G{\mathbf{v}^{k+1}}$ to $\G{\mathbf{v}^{k}}$.
\begin{figure}
\begin{center}
\includegraphics[width=0.9\textwidth,clip, trim = 2.3cm 1cm 2.4cm 1cm]{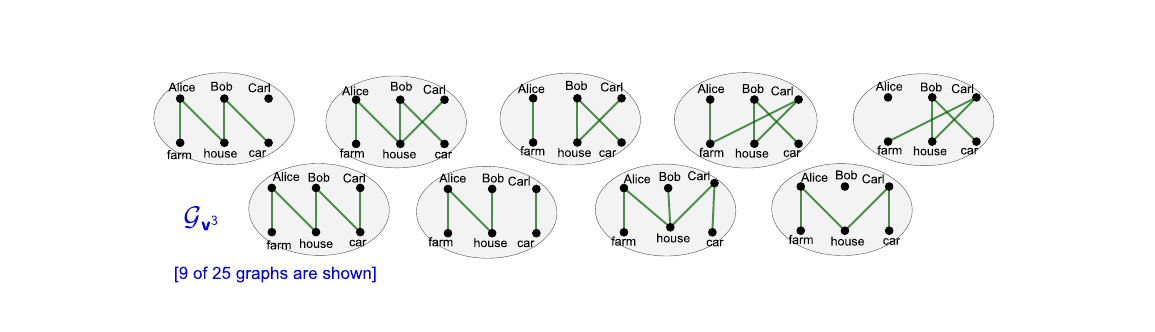}
\end{center}
\caption{
\label{fig:cg-tree-level3}
Enumerating the fPO consumption graphs of allocations among three agents (Lemma \ref{lem:enumerate-levelk}) in a non-degenerate instance with $m=3$ pure goods. 
Shows some of the consumption graphs in $\G{v^{3}}$ derived from the top-right consumption graph in Figure \ref{fig:cg-tree-level2}.
The graphs in the top row are derived by sharing Alice's goods (the farm and the house) with Carl;
the graphs in the second row 
are derived by sharing Bob's goods (the car and the house) with Carl.
}
\ifdefined\Description
\Description[Enumerating the fPO consumption graphs of allocations among three agents]{Enumerating the fPO consumption graphs of allocations among three agents}
\fi
\end{figure}
\begin{lemma}
	\label{lem:enumerate-levelk}
	For   a valuation $\mathbf{v}$ with $n\geq 3$ agents and  $k$ such that $2\leq k\leq n-1$:
	
	(a)
	The number of graphs in $\G{v^{k+1}}$ satisfies the upper bound
	\begin{equation}\label{eq:dyn-prog_upper_bound}
	\big|\G{v^{k+1}}\big|\leq \big|\G{v^{k}}\big|\cdot 3^{(\D{v}+1)k}\cdot m^k
	\end{equation}

(b) All the graphs  in $\G{v^{k+1}}$  can be enumerated using 

$$O\left(\big|\G{v^k}\big|\cdot  3^{\D{v}\cdot k} \cdot m^{2+k}\right)$$  
	operations if $\G{v^{k}}$ is given as the input.
\end{lemma}
Since $\mathbf{v}^n=\mathbf{v}$, starting from $\mathbf{v}^2$ (covered by Lemma~\ref{lem:enumerate-level2}) and repeatedly applying Lemma~\ref{lem:enumerate-levelk}  we get both the algorithmic part of Proposition~\ref{prop:nagents-po} as well as the upper bound~\eqref{eq_number_PO_ucg}.
\begin{proof}[Proof of Lemma~\ref{lem:enumerate-levelk}.]
The idea is that any graph $G'\in \G{v^{k+1}}$ can be obtained from some $G\in \G{v}$ by erasing some of the edges between objects and ``old'' agents $i\in[k]$ and tracing new edges to a ``newcomer'' $k+1$ in such a way that for each old agent $i$, the allocation of objects between $i$ and the newcomer is fPO in the $2$-agent subproblem.
	
First, we check that any fPO allocation $\mathbf{z'}$ among $k+1$ agents can be obtained from an fPO allocation $\mathbf{z}$ among $k$ by reallocating some objects to the newcomer. Indeed, by Lemma~\ref{lem:po-weights}, there exists a vector of weights $\lambda=(\lambda_i)_{i\in[k+1]}$ with strictly positive components such that in allocation $\mathbf{z'}$ each object $o$ is consumed by agents $i$ with a highest $\lambda_i v_{i,o}$. Giving the share $z_{k+1,o}'$ of each object $o$ consumed by $k+1$ to an agent $i\in[k]$ with a highest $\lambda_i v_{i,o}$ defines the desired allocation $\mathbf{z}$; fPO follows from the same Lemma~\ref{lem:po-weights} with vector $(\lambda_i)_{i\in[k]}$.

Second, we describe how the reallocation looks in terms of graphs. For $G\in \G{v^k}$, consider the set of objects consumed by agent $i$: $A_i(G)=\{o\in [m]\ :\ \mbox{there is an edge between $i$ and $o$}\}$. For each $i$ consider a $2$-agent problem $\mathbf{v}^{i,k+1}(A_i)$, where $i$ and $k+1$ divide the set $A_i(G)$ of objects between themselves. Pick a graph $G^{i,k+1}\in \G{{v}^{i,k+1}(A_i)}$ for each $i\in [k]$, which prescribes how objects in $A_i(G)$ are reallocated between $i$ and $k+1$.
Define the graph $G'$ as follows: agent $i\in[k]$ and object $o$ are connected by an edge if they are connected in $G^{i,k+1}$; an edge between $k+1$ and $o$ is traced if this edge exists in $G^{i,k+1}$ for at least one agent $i\in[k]$.

Denote by $\G{v^{k+1}}'$ the set of graphs $G'$ that we get, when $G$ ranges over $\G{v^k}$ and $G^{i,k+1}$ over $ \G{{v}^{i,k+1}(A_i)}$ for all $i\in [k]$. By the construction, $\G{v^{k+1}}'$ contains all the consumption graphs of fPO allocations for $\mathbf{v}^{k+1}$, but may contain some non-fPO graphs, since the reallocation preserves the fPO condition only for pairs of agents $i,j\in[k+1]$. In order to get $\G{v^{k+1}}$, each graph $G\in \G{v^{k+1}}'$ must be tested for fPO using Lemma~\ref{lem:po-check} and those graphs that do not pass the test must be eliminated.

(a)
Let us estimate the total number of graphs in $\G{v^{k+1}}$. For each  $G\in\G{v^{k}}$, the set $\G{v^{i,k+1}(A_i)}$ contains at most $|A_i(G)|\cdot 3^{\D{v}+1}$ graphs (see Lemma~\ref{lem:enumerate-level2}). Therefore, the total number of graphs $G'$ obtained from $G$ is bounded by $\left(3^{(\D{v}+1)}\right)^k\cdot \prod_{i=1}^k |A_i(G)|\leq \left(3^{(\D{v}+1)}\cdot m\right)^k$ and we get~\eqref{eq:dyn-prog_upper_bound}.

 
(b)
The bound on time-complexity follows from Lemma~\ref{lem:enumerate-level2} as well. For each $G$ and $i\in[k]$, computing $\G{v^{i,k+1}(A_i)}$ takes $O(m\cdot 3^{\D{v}})$ operations if prior to that for each pair of agents, $i\in [k]$ and $k+1$, objects with non-zero values are ordered by $\frac{v_{i,o}}{v_{k+1,o}}$. Thus, all $G'$ for a given $G$ are enumerated in time $O\left(\left(m\cdot 3^{\D{v}}\right)^k\right)$ and the time needed for reordering the objects is absorbed by this expression.
Checking fPO takes additional $O(m^2)$ for each $G'$ by Lemma~\ref{lem:po-check}. Thus, the overall time complexity is $O\left(\big|\G{v^k}\big|\cdot m^{2+k}\cdot 3^{\D{v}\cdot k} \right)$%
.
\end{proof}


\section{Implementation and Experiments.}\label{sec:simulations}
\subsection{Practical Considerations.}\label{subsect:practical_cons}
The algorithm from Theorem \ref{thm:nagents-different-easy} was implemented in Python.
The code was written by Eliyahu Satat (an undergraduate student at Ariel University); see
\url{https://github.com/erelsgl/fairpy/blob/master/fairpy/items/min_sharing.py}.

In preliminary experiments with random instances, we found out that many instances have fewer than the upper bound of $n-1$ sharings. To take advantage of this finding, we implemented the following variant of the algorithm:

\begin{mdframed}
\begin{minipage}{\textwidth}
\begin{itemize}
\item For $s := 0, \ldots, n-1$:
\begin{itemize}
\item Run the algorithm of Theorem \ref{thm:nagents-different-easy} with an upper bound of {$s$} on the number of sharings.
I.e., a descendant of the current consumption graph is not explored if it has more than $s$ sharings.
\end{itemize}
\end{itemize}
\end{minipage}
\end{mdframed}

Essentially, we first look for allocations with $0$ sharings, then with $1$ sharing, etc. 
This does not affect the worst-case runtime, but speeds up the algorithm in practice whenever the instance admits an allocation with fewer than $n-1$ sharings. Indeed, 
the algorithm avoids enumerating all the fPO consumption graphs and computes only those that have at most $s$ sharings.

Additionally, we stop exploring the descendants of the current graph if it cannot lead to a proportional allocation. This is easy to check: for each agent $i$, calculate the sum of values of all goods that are adjacent to $i$ in the graph. If the sum is less than $\frac{1}{n}\sum_{o\in[m]} v_{i,o}$, then no descendant of this graph can result in a proportional allocation.


\subsection{Experiments.}

We ran our algorithm on real fair division instances with goods from \url{spliddit.org} \citep{Goldman2015Spliddit}, 
which were kindly shared with us by Nisarg Shah. 
Spliddit can be used either in non-demo mode (users create the instance) or in demo mode (the platform suggests an instance that can be then adjusted by users). We focused on non-demo instances only:
the database contained $717$ such instances that were recorded as of 7/2020.
We restricted our analysis to $703$ instances with at most $8$ agents.
For each instance in this set, we computed the minimal number of sharings in an allocation that is PROP and fPO or EF and fPO.

\subsubsection*{CEEI as a benchmark.} We compared the outcome of our algorithm with the benchmark of competitive equilibrium with equal incomes (CEEI), the most popular way to divide divisible goods; see Section~\ref{sec:intro} and~\ref{sec:related}. We computed CEEI via the Eisenberg-Gale convex program (maximization of the product of agents' utilities).	Typically, the optimum is unique and corresponds to the unique CEEI allocation; however, for degenerate instances the set of optima may not be a singleton, and each such optimum is a CEEI. In the case of non-uniqueness, we selected a CEEI allocation with the minimal number of sharings. To conduct this minimization, we relied the fact that, by the convexity of the Eisenberg-Gale program, all the optima  result in the same vector of utilities~$(w_i)_{i=1}^n$.  As CEEI is fPO, the minimization over CEEI is equivalent to the minimization over the set of allocations, where each agent $i$'s utility is at least $w_i$. The latter problem was solved by a straightforward  modification of our sharing-minimization algorithm for PROP with the lower bound $\frac{1}{n}\sum_{o\in[m]}v_{i,o}$  on agent $i$'s utility  replaced by 
$(1-\epsilon)w_i$,
where $\epsilon=0.001$ was chosen to avoid floating-point rounding errors.

Note that CEEI is contained in the set of fPO+EF allocations, and the set of fPO+EF allocations is a subset of fPO+PROP ones. The first inclusion is typically strict since, for generic instances with goods, CEEI picks a single allocation, while  fPO+EF is not a singleton. For $n=2$, EF is equivalent to PROP, but for $n\geq 3$, the set of fPO+EF allocations is typically a strict subset of  fPO+PROP.
Thus the minimal number of sharings weakly increases if we replace the requirement of fPO+PROP  by fPO+EF and fPO+EF by CEEI. For some instances, all these  inequalities can be strict.  Consider the following example with three agents and four goods: 
\begin{align*}
\mathbf{v} := \begin{bmatrix}
10 && 18 && 1 && 1
\\
10 && 18 && 1 && 1
\\
10 && 10 && 5 && 5
\end{bmatrix}.
\end{align*}
Running the algorithm of sharing minimization under the constraint of fPO+PROP  / fPO+EF / CEEI yields allocations with 0 / 1 / 2 sharings, respectively:
\begin{align*}
\mathbf{z}^{PROP} =
\begin{bmatrix}
1 && 0 && 0 && 0
\\
0 && 1 && 0 && 0
\\
0 && 0 && 1 && 1
\end{bmatrix}
&&
\mathbf{z}^{EF} =
\begin{bmatrix}
1 && 0 && 0 && 0
\\
0 && 0.556 && 0 && 0
\\
0 && 0.444  && 1 && 1
\end{bmatrix}
&&
\mathbf{z}^{CEEI} =
\begin{bmatrix}
0.734 && 0.296 && 0 && 0
\\
0 && 0.704 && 0 && 0
\\
0.266 && 0  && 1 && 1
\end{bmatrix}.
\end{align*}

\subsubsection*{Conclusions from the Spliddit data.}
We let the sharing-minimization algorithm run for up to $T = 1000$ seconds for each instance on a laptop HP EliteBook 840 G3. For fPO+PROP and fPO+EF,	almost all instances with at most $4$ agents were completed within this time frame. For instances that did not complete in time, the algorithm output an allocation with $n-1$ sharings (the worst-case upper bound).
Figure \ref{fig:runtime} shows how the runtime depends on the number of agents for different types of constraints. We observe that the runtime increases hyper-exponentially in the number of agents since, in the Spliddit data, instances with more agents typically have more objects. For $5$ or more agents, we observe that sharing minimization takes significantly more time for fPO+EF than for fPO+PROP. The unexpected decrease in the runtime for fPO+PROP with $7$ and $8$ agents is explained below.

\begin{figure}
\begin{center}
\includegraphics[width=0.7\textwidth]{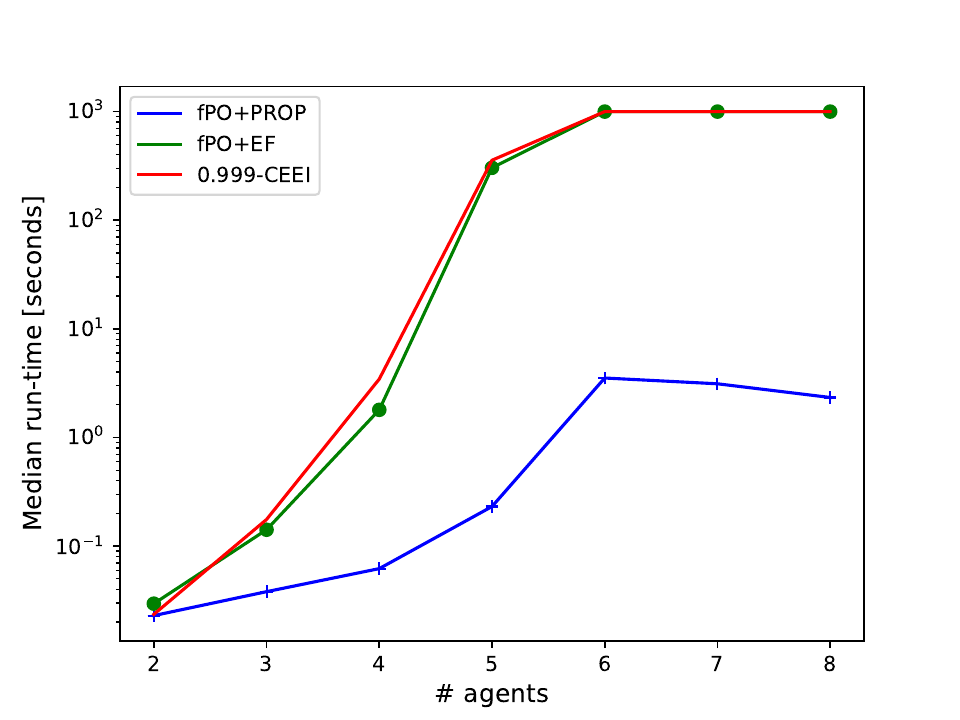}
\end{center}
\caption{\label{fig:runtime}
Median runtime of our algorithm (timed out at 1000 seconds) as a function of the number of agents.
}
\end{figure}

\begin{figure}
\begin{minipage}[c]{0.8\textwidth}
\includegraphics[width=\textwidth]{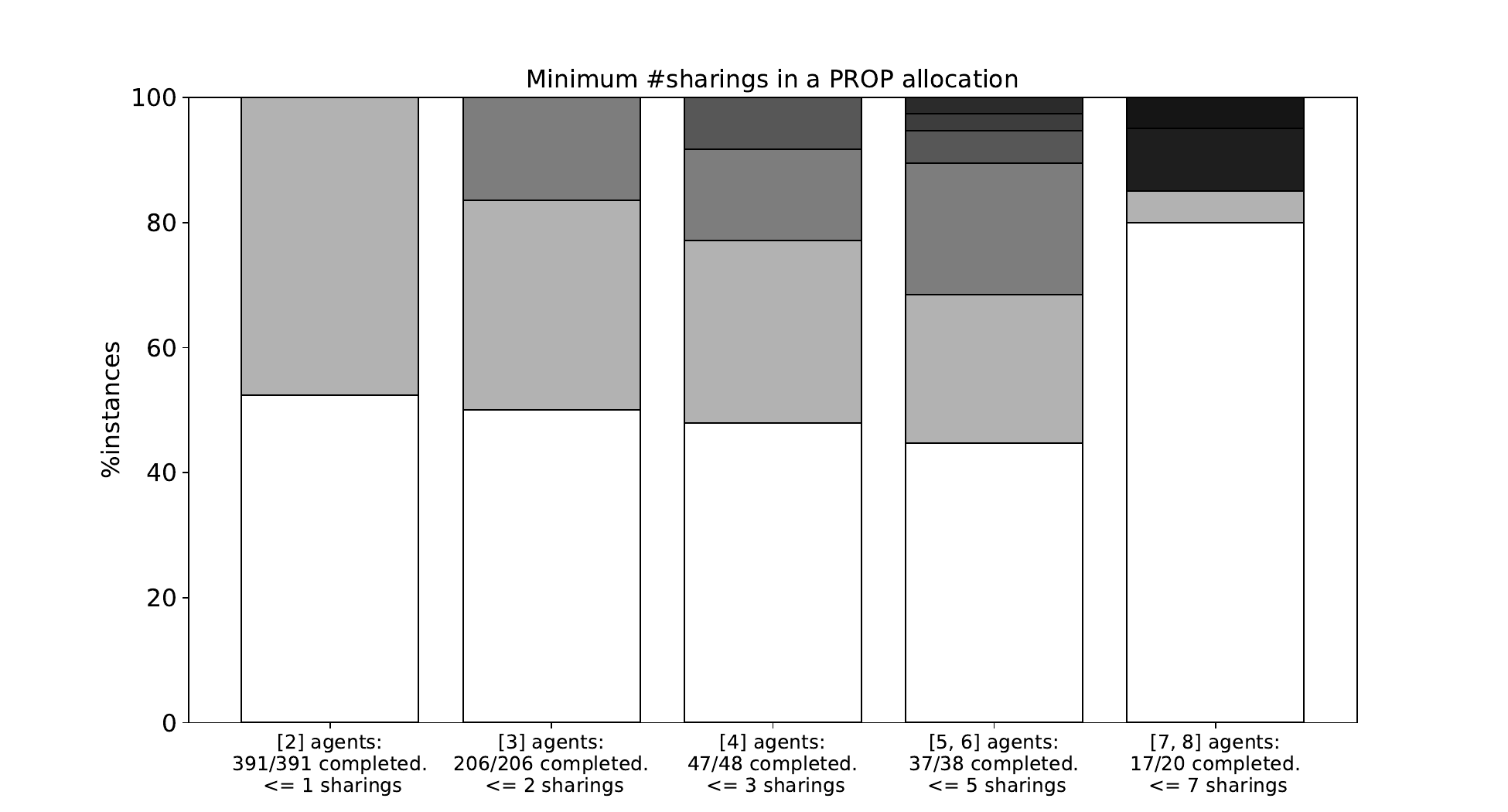}
~~
\includegraphics[width=\textwidth]{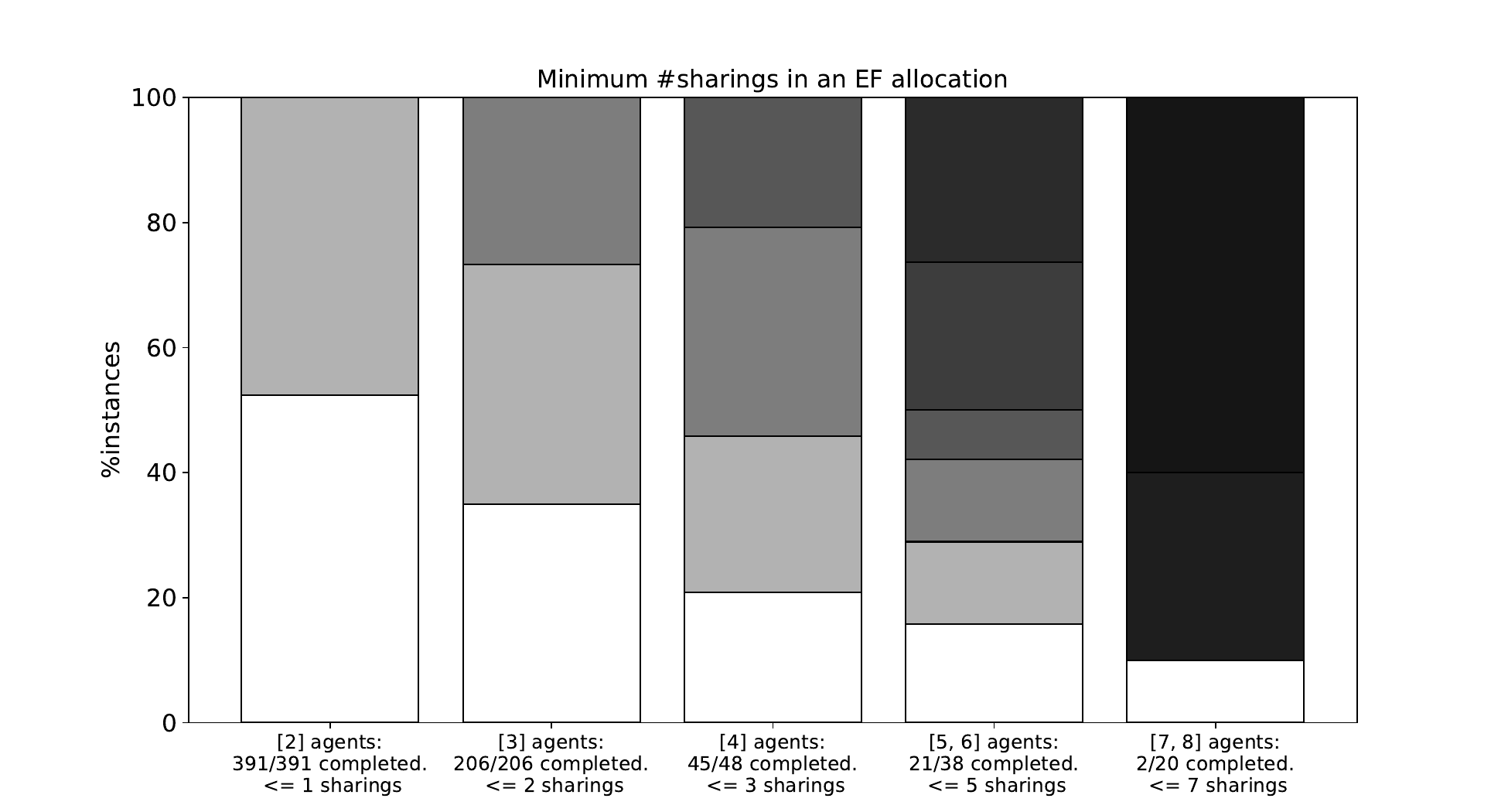}
~~
\includegraphics[width=\textwidth]{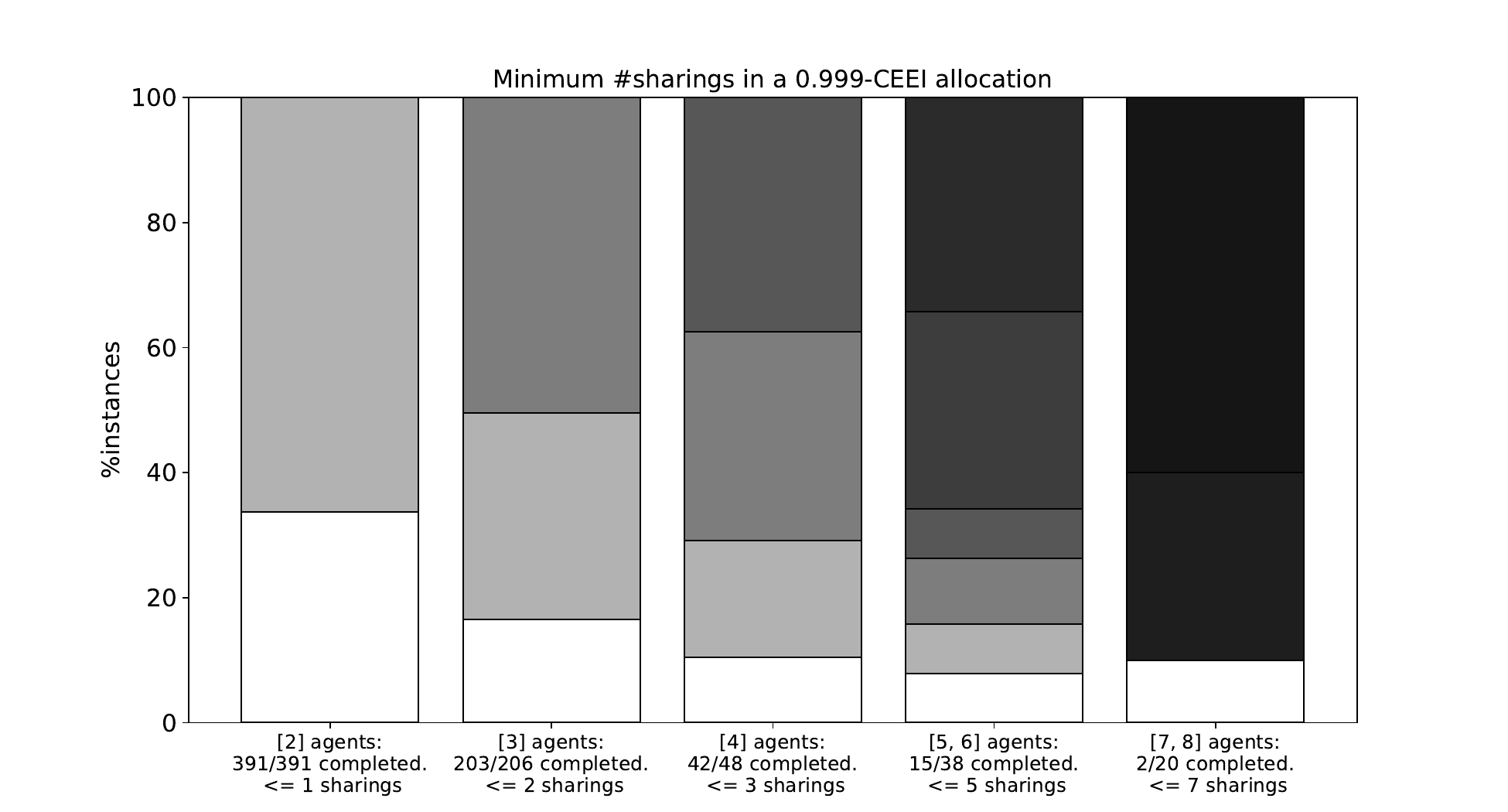}
\end{minipage}
\hfill
\vskip 1cm
\caption{
\label{fig:spliddit-sharing}
Results of experiments with Spliddit instances. See explanations in text.
}
\end{figure}

Figure \ref{fig:spliddit-sharing} represents the percentage of instances with a given 
minimal numbers of sharings, for different numbers of agents and different constraints (fPO+PROP, fPO+EF, or CEEI).
The white regions at the bottom of the bars represent the fraction of instances that admit an allocation with $0$ sharings. The adjacent light-gray regions correspond to the fraction of instances that require $1$ sharing; the next darker regions corresponds to the fraction with $2$ sharings, and so on. 

We observe several phenomena. Most of the instances admit a fair fPO allocation with fewer sharings than the worst-case bound of $n-1$. For example, approximately $50\%$ instances with $n=2$ agents have a proportional allocation with no sharings at all; for $n=3$ and $n=4$, the fraction of instances requiring $n-1$ sharings is below $20\%$ and $5\%$, respectively. Moreover, the minimal number of sharings in an fPO+PROP and fPO+EF allocations are often smaller than in CEEI; see Table \ref{tab:compare-sharings} for the exact  numbers.
Interestingly, many Spliddit instances with $7$ or more agents admit an fPO+PROP allocation with no sharings at all. Because of this, the algorithm does not explore allocations with several sharings and terminates fast, which explains the counter-intuitive decrease in the runtime for fPO+PROP in Figure~\ref{fig:runtime}.

We conclude that real fair division instances often require much fewer sharings than in the worst case. This opens room for optimization.  CEEI improves upon the worst case, and further improvements are attained by our sharing-minimization approach. This indicates the potential of sharing minimization for increasing  participants' satisfaction in practice.

\begin{table}
\begin{center}
\begin{tabular}{|c|c|c|}
\hline
 & Fewer sharings than worst-case & Fewer sharings than CEEI \\
\hline
fPO+PROP & 475/703=67=67.6\% & 255/703=36.2\% \\
\hline
fPO+EF & 416/703=59.1\% & 166/703=23.6\% \\
\hline
\end{tabular}
\end{center}
\caption{
\label{tab:compare-sharings}
Percentage of Spliddit instances 
(with at most 8 agents) 
in which the minimal number of sharings in an fPO+PROP or fPO+EF allocation is smaller than in the worst-case ($n-1$) or in CEEI.
}
\end{table}

\section{Related Work.}
\label{sec:related}

The classic economic approach to fair division assumes that resources are divisible, preferences are  general (e.g., arbitrary convex as in the theory of general equilibrium), and aims to identify   division rules that satisfy a family of requirements known as axioms \citep{thomson2011fair}; this approach often leads to impossibility results as the axioms turn out to be incompatible. A recent trend 
largely inspired by researchers with CS and AI background is to escape impossibilities by narrowing down the preference domain (most papers deal with additive utilities) and substituting the exact axioms by their quantitative relaxation, thus replacing the question of the existence  of an ideal rule by the problem of finding the best approximation to the ideal \citep{brandt2016handbook}. Our paper follows this modern methodology.  A survey by \cite{moulin2019fair} touches both approaches, discusses the role of market equilibria in fair division, contrasts divisible resources with indivisible objects, and goods with bads. We recommend it as a short but comprehensive introduction to modern fair division.
Below, we describe several strains of the literature closest to our results and focus primarily on the case of additive utilities.

\subsection{Fisher markets, CEEI, and the Eisenberg-Gale program for goods and bads}
Since the seminal paper by \cite{Varian1974Equity}, fair division has been closely connected to the theory of competitive  equilibria (CE) of exchange economies known in computer science literature as Fisher markets. Varian demonstrated that a CE in the economy where all the agents have equal endowments (CE with equal incomes or simply CEEI) is both envy-free and fractionally Pareto-optimal. Therefore, the existence of CE ---  known in the theory of general equilibrium under relatively weak assumptions of continuity, convexity, and monotonicity of preferences \citep{arrow1954existence} --- translated into the existence of fair and efficient allocations.

Under the additional assumption that the utilities are homogeneous (e.g., additive, Leontief, CES, or Cobb-Douglas), \cite{eisenberg1959consensus} showed that CEEI maximizes the product of agents' utilities, the so-called Nash product; see also \citep[Chapters 5 and 6]{nisanalgorithmic}. The convexity of this optimization problem implies  the  uniqueness of CEEI in the space of utilities as well as the approximate algorithm for computing CEEI via  standard gradient methods. For additive utilities, \cite{orlin2010improved} and \cite{garg2019strongly} developed exact combinatorial algorithms running in strongly polynomial time; their constructions implicitly use the convexity of the problem. These algorithms can also be seen as algorithms for computing envy-free fractionally Pareto optimal allocations, and, surprisingly, no algorithms except market-based are known for this problem. 

\cite{Bogomolnaia2017Competitive} considered CEEI for preferences that are homogeneous and continuous but not necessarily monotone, which includes the case of a mixture of goods and bads under additive utilities. CEEI exists and remains envy-free and fractionally Pareto-optimal. It is still related to the Nash product of the absolute values of utilities, but this connection becomes more subtle; we explain this subtlety in the case of pure bads and additive utilities. Global extrema of the product do not correspond to equilibrium allocations --- the minimum is attained at an unfair allocation, where one of the agents gets no bads, while the maximum lies on the anti-Pareto frontier. It turns out that CEEI coincides with the local critical points of the Nash product on the Pareto frontier (local minima, maxima, and saddle points). In particular, CEEI for bads does not solve a convex optimization problem, and there can be an exponential number of distinct CEEI. Despite this obstacle, CEEI for bads can be computed in strongly polynomial time if either the number of agents or the number of bads is fixed \citep{branzei2019chores}. The result extends to a mixture of goods and bads \citep{garg2020computing} and to some classes of non-additive utilities \citep{chaudhury2021competitive}.

Usefulness of \emph{agent-object graphs} --- directed bipartite graphs in which the nodes one one side are the agents and the nodes on the other side are the objects --- was recognized by many algorithmic and economic applications. See for example \citet{cole2016convex,barman2018proximity}.

\subsection{Known Worst-case and Average-case Bounds on Sharing.}

The idea of finding fair allocations with a bounded number of shared goods originated from Brams and Taylor \citep{Brams1996Fair,brams2000winwin}.
They suggested the \emph{Adjusted Winner (AW)} procedure, which finds fair and fractionally Pareto-optimal allocation of goods between two agents with additive utilities and at most~$1$ sharing, i.e., with the worst-case optimal number. The AW procedure was applied (at least theoretically) to division problems in divorce cases and international disputes \citep{Brams1996Camp,Massoud2000Fair} and was studied empirically \citep{Schneider2004Limitations,Daniel2005Fair}. The AW procedure heavily relies on the simple structure of fPO allocations for two agents (see the proof of Theorem~\ref{thm:2agents-different-easy} and \cite{moulin2004fair}, Example~7.11a).
Brams and Taylor do not extend their AW procedure to three or more agents.

For $n\geq 3$ agents, the number of sharings was studied in an unpublished manuscript of \citet{wilson1998fair}. 
He proved worst-case bounds on sharing for fairness criteria that may be incompatible with fractional Pareto-efficiency. For example, 
he proved the existence of an \emph{egalitarian}
allocation of goods 
 --- an allocation in which all agents have a largest possible equal utility \citep{pazner1978egalitarian} --- with $n-1$ sharings. Egalitarian allocations of goods are proportional but not necessary envy-free, and may violate efficiency if the valuation matrix has zeros.
For such criteria, the approach based on our Lemma \ref{lem:po-n-1} becomes inapplicable;  Wilson uses a different technique based on linear programming.

Recently, \citet{goldberg2020consensus} studied the problem of minimizing sharing in \emph{consensus halving} --- a partition of  objects into two subsets each of which has a value of exactly half for all agents.

There is a significant gap between worst-case and average-case numbers of sharings; this further stresses the importance of sharing minimization. \citet{Dickerson2014Computational} considered  random instances and demonstrated
 that the minimal number of sharings for an envy-free fPO allocation is zero with high probability, if the number of goods is large and values are independent and identically distributed.
 \citet{ManurangsiSu17} extended the result to allocations among agent groups.


\subsection{Fairness with Indivisible Objects.}
With indivisible objects, envy-free and even proportional allocations may not exist. The {most commonly studied} relaxations of these two concepts are  \emph{Envy-freeness up to one good (EF1)} and \emph{Maximin share guarantee (MMS)}.

EF1 was introduced by \citet{budish2011combinatorial}; a closely related concept was presented earlier by \citet{Lipton2004Approximately}. 
EF1 is widely studied, for example by \citet{Aleksandrov2015Online,oh2019fairly} and others. 
A strengthening of EF1 to a global fairness notion, based on information withholding, was recently studied by \citet{hosseini2020fair}.
Existence of EF1 indivisible Pareto Optimal allocations was proved by \citet{caragiannis2019unreasonable} and \citet{barman2018finding} strengthened the result to fractional PO. 

\citet{budish2011combinatorial} defined MMS, demonstrated existence under large-market assumption and applied the concept in practice for course allocation in \citep{budish2016course}. For ``small markets'', \citet{procaccia2018fair} showed that MMS allocations may not exist for some knife-edge instances and hence all the results about MMS consider a certain approximation to MMS itself, e.g.,~
\citep{aziz2016approximation,amanatidis2017approximation,barman2017approximation,ghodsi2018fair, garg2018approximating,babaioff2019competitive,segal2020competitive,babaioff2021competitive}. 

Classic microeconomics mostly works with divisible resources and handles indivisible objects by making them ``divisible'' via a lottery. This approach results in weaker fairness: an allocation is fair ex-ante, i.e., in expectation before the lottery is implemented. Ex-ante fairness was analyzed in many different contexts. Just to name a few:  \citet{hylland1979efficient}, \citet{abdulkadiroglu1998random},  \citet{Bogomolnaia2001New} considered the problem of fair assignment  and  \citet{budish2013designing} considered its multi-unit constrained  modifications; \citet{kesten2015theory} evaluated fairness of tie-breaking in matching markets; \citet{bogomolnaia2019simple} studied ex-ante fairness under additional randomness in preferences.

\citet{Brams2013TwoPerson} observed that exact envy-freeness with indivisible goods can be achieved by leaving some of them unallocated while keeping some efficiency guarantees: their AL procedure constructs an allocation  that is not Pareto-dominated by another envy-free allocation, see also~\citet{Aziz2015Generalization}.
Recently, \citet{caragiannis2019envy} showed that, by leaving some items unallocated,  it is possible to construct an allocation that is ``Envy-free up to any good'' (also known as EFx, an approximate fairness notion which strengthens EF1 \citep{caragiannis2019unreasonable}).

\citet{bei2020fair} studied an allocation problem where some goods are divisible, and some are indivisible. In contrast to our setting, in their paper, the partition of the set of objects into divisible and indivisible is given in advance --- the algorithm can only divide objects that are predetermined as divisible.

\citet{halpern2019fair}
 suggested a novel approach to achieve exact fairness with indivisible goods: introduce a  \emph{monetary subsidy} by a third party, while minimizing the transfers. This approach is further developed in the follow-up papers~\citep{brustle2019one,caragiannis2020computing}. Minimization makes it methodologically similar to our approach and distinguishes it from other results on fair allocation of indivisible goods with monetary transfers, e.g., the rent-division problem \citep{Su1999Rental,gal2017fairest}.
In practical cases of real-estate allocation,
it is common to allow both sharing and side-payments, 
and the goal is to minimize both of them
 \citep{nizri2021private}. 
No efficient solution to this problem seems to be currently known.

Another way of bridging the literature on divisible and indivisible objects was proposed by \citet{freeman2020best} and \citet{aziz2020simultaneously}. They interpret fractional allocations as  randomization over integral ones and distinguish  ex-ante fairness of a fractional allocation and ex-post fairness of an indivisible allocations from the support of the lottery. The papers provide algorithms for computing an allocation that is ex-ante envy-free and ex-post envy-free up to one good (without any guarantees of economic efficiency).
The fractional allocation generated by our algorithm can naturally be interpreted as a probabilistic allocation. Viewed like that, 
minimizing the number of shared objects is equivalent
to minimizing the number of objects whose allocation is affected by randomization, subject to being ex-ante fair and fPO.
However, the resulting allocations may not be ex-post fair.
Minimization of the number of randomized objects subject to both ex-ante and ex-post fairness is an interesting open problem.

\subsection{Checking Existence of Fair Allocations of Indivisible Objects.}
Fair allocation of indivisible goods might not exist in all cases, but may exist in some. A natural question is how to decide whether it exists in a given instance.
It was studied by \citet{Lipton2004Approximately,deKeijzer2009Complexity,bouveret2016characterizing} for various fairness and efficiency notions, showing that it is computationally hard in general (with some exceptions).  The \emph{undercut procedure}  of \citet{Brams2012Undercut}  finds an envy-free allocation of indivisible goods among two agents with monotone (not necessarily additive) valuations, if-and-only-if it exists (see also \citet{Aziz2015Note}).

\subsection{Cake-cutting with Few Cuts.}
The goal of minimizing the number of ``cuts'' has also been studied in the context of \emph{fair cake-cutting} --- dividing a heterogeneous continuous resource, see 
\citet{Webb1997How,Shishido1999MarkChooseCut,Barbanel2004Cake,Barbanel2014TwoPerson,alijani2017envy,seddighin2018expand,SegalHalevi2018Entitlements,crew2019disproportionate}. 
Since the resource is continuous, the techniques and results are quite different.

\subsection{Fair Division with Mixed Valuations.}
Most of the literature on fair division deals with either goods or bads. Recently, some papers have studied mixed valuations, where objects can be good for some agents and bad for others. This setting was first studied by \citet{Bogomolnaia2017Competitive} for divisible objects and quite general class of utilities.
It was later investigated by  
\citet{segal2018fairly,meunier2018envy,avvakumov2019envy} for a heterogeneous divisible ``cake'',
and by \citet{aziz2018fair,aleksandrov2020jealousy,aziz2020polynomial,aleksandrov2020two} for indivisible objects and approximate fairness.

\section{Extensions and Future Work}
The sharing-minimization approach introduced in this paper provides a compelling alternative to existing approaches for those fair division problems, where objects can technically be shared, but this sharing is unwanted as happens in many practical cases. Our approach guarantees exact fairness and hence improves upon approximate fairness (unavoidable if objects are treated as indivisible) and often produces fewer sharings than benchmark  rules  
for divisible objects such as CEEI.

The idea behind our sharing-minimization algorithm from Theorem~\ref{thm:nagents-different-easy} is quite general. The algorithm can easily be adapted for optimizing other objectives over the set of fPO allocations provided that, for a given consumption graph, the optimization reduces to an LP. For example, one can minimize the number of shared objects, the number of sharings in a given subset of objects $S\subset [m]$ (where $S$ is a subset of objects that are particularly ``hard to share''), the total value of shared objects $\sum_{i\in[n]}\sum_{o \  \mathrm{shared \ by } \ i}|v_{i,o}|$, or the maximum number of agents who share a single object.
For these objectives, 
 step \ref{item:upperbound} in the algorithm is to be skipped (i.e., fPO graphs with more than $n-1$ sharings cannot be discarded).
Still, by Remark \ref{rem:sharing-upper-bound}, the number of sharings in any fPO allocation is at most 
$\frac{n(n-1)}{2}(\D{v}+1)$.
Hence, the total number of variables and constraints in the LP from step~\ref{item:lp} is independent of $m$, the LP is solvable in constant time (when $n$ is fixed), and the complexity is dictated by the time needed to enumerate all fPO graphs, i.e., the bound on the run time from item (a) of Theorem~\ref{thm:nagents-different-easy} holds.

Similarly, instead of envy-freeness one can use other fairness notions, such as weighted envy-freeness
\citep{reijnierse1998finding,branzei2019chores} or weighted proportionality capturing different entitlements, 
or any other fairness notion that can be represented by a constant (independent of $m$) number of linear inequalities on the allocation matrix.
If there exists an fPO allocation satisfying the chosen fairness notion, our algorithm will find it. 
Otherwise, the algorithm will indicate that such an allocation does not exist.
\medskip 

Some directions for future research that we find promising:
\begin{itemize}
\item We expect that there is room for substantial improvement of the algorithm's run time in practice. Presumably, the theoretical exponent in the run time can also be improved by replacing the exhaustive enumeration of the Pareto frontier by a more clever procedure not exploring regions that cannot contain the optimum; see Section~\ref{subsect:practical_cons} for some ideas along these lines. Understanding such improvements can make the algorithm more practical for many-agent problems. If the number of agents exceeds the number of objects, an efficient algorithm can presumably be obtained by invoking agent-object parity from \citep{branzei2019chores}.   
\item When sharing is especially problematic, further improvement in the number of sharings can be achieved by tweaking the efficiency requirement of fPO. 
For example, one can allow for allocations with bounded inefficiency quantified using the loss in utilitarian or Nash social welfare, or by replacing the requirement of  fractional PO by discrete PO. For most of such modifications, sharing minimization is likely to becomes computationally hard, and it is especially interesting to find modifications for which it does not.
A related direction is finding domains of preferences beyond additive valuations, where sharing minimization remains tractable and/or there is a  non-trivial worst-case  bound on the number of sharings.
\item An alternative more economic approach to sharing-minimization would be to assume that sharing causes disutility to agents, thus making the sharing-minimization objective hardwired in economic efficiency.
The biggest challenge is to find a  domain of sharing-averse preferences, where the existence of a fair efficient allocation is guaranteed.
\end{itemize}

\ACKNOWLEDGMENT{
This work was inspired by Achikam Bitan, Steven Brams and Shahar Dobzinski, who expressed their dissatisfaction with the current trend of approximate-fairness (SCADA conference, Weizmann Institute, Israel 2018). We are grateful to participants of De Aequa Divisione Workshop on Fair Division (LUISS, Rome, 2019), Workshop on Theoretical Aspects of Fairness (WTAF, Patras, 2019) and the rationality center game theory seminar (HUJI, Jerusalem, 2019) for their helpful comments.  Suggestions of Herve Moulin, Antonio Nicolo, and Nisarg Shah were especially useful. We also thank Nisarg for sharing with us the data from \url{spliddit.org}.

We are grateful to Eliyahu Satat and Daniel Abergel, undergraduate students at Ariel University, for the Python implementation of our algorithm at 
\url{https://github.com/DanielAbergel/Distribution-Algorithm}, and for their help in using the code.

Several members of the theoretical-computer-science stack-exchange network (\url{http://cstheory.stackexchange.com}) provided very helpful answers, in particular: 
D.W, Peter Taylor, 
	Gamow, Sasho Nikolov,
	Chao Xu, Mikhail Rudoy
	and xskxzr%
	. 
	
	We are grateful to anonymous referees of WTAF 2019, EC 2020 and Operations Research journal for their very helpful comments. We are grateful to Lillian Bluestein for proofreading the early version.
	
	Fedor's work was supported by the Ministry of Science and Technology grants \#19400214 and \#2028255, the Lady Davis Foundation, 
	Grants 16-01-00269 and 19-01-00762 of the Russian Foundation for Basic Research,  the European Research Council (ERC) under the European Union's Horizon 2020 research and innovation program (grant agreement n$\degree$740435), and the Basic Research Program of the National Research University Higher School of Economics.
	Erel is supported by the Israel Science Foundation (grant 712/20).
}

\section*{Biographies}
\begin{description}
\item[\textbf{Fedor Sandomirskiy}]  is a postdoc at Caltech. His specialization is in game theory and its applications to problems of economic design, with a focus on the strategic use of information and resource allocation mechanisms. Fedor earned a PhD from the Russian Academy of Sciences and worked as a postdoc at Technion before joining Caltech. His research was recognized by  EC2020 best paper award, Ovsievich and Deych prizes, and multiple grants.

\item[\textbf{Erel Segal-Halevi}] 
has started to study fair division  after reading the Biblical commandment to ``divide the land equally'' (Ezekiel 47:14).
This study has lead him to do a Ph.D. in computer science in the Bar-Ilan University (2013--2017), focusing on algorithms for fair division of land.
He is currently a faculty member in the Ariel University, where he is working on new fair division algorithms, as well as teaching his students how to use the existing ones.
\end{description}

\newpage

\ifdefined\OPRE
\begin{APPENDICES}
\else
\appendix
\fi

\section{Characterization of Fractional Pareto-Optimality.}
\label{sec:po-cycle}

In this section we prove Lemma \ref{lem:po-cycle} and Lemma~\ref{lem:po-weights} together:
\begin{itemize}
\item \emph{An allocation $\mathbf{z}$ is fractionally Pareto-optimal if-and-only-if 
it is non-malicious and its directed consumption graph $\dcg{z}$ has no cycle $C$ with $\pi(C)<1$.}
\item \emph{An allocation $\mathbf{z}$ is fractionally Pareto-optimal if and only if there is a vector of weights  $\lambda=(\lambda_i)_{i\in [n]}$ with $\lambda_i>0$ such that for all agents $i\in[n]$ and objects $o\in[m]$:}
\begin{equation}
\label{eq:lm-weights-appendix}
z_{i,o}>0 \ \ \mbox{implies} \ \  \lambda_i v_{i,o}\geq \lambda_j v_{j,o} \ \ \mbox{ for any agent} \ \ j\in [n].
\end{equation}
\end{itemize}
\begin{proof}[Proof that ${fPO}$ $\implies$ no $C$ and no maliciousness]
If an allocation is malicious, then reallocating objects in a non-malicious way strictly improves the utilities of some agents without harming the others. Thus, an fPO allocation $\mathbf{z}$ must be non-malicious.

We now show that there are no directed cycles $C = (i_1 \to o_1 \to i_2 \to o_2 \to\allowbreak \ldots \to \allowbreak i_L \to o_L \to i_{L+1} = i_1)$  in $\dcg{z}$ with $\pi(C)<1$.
 Assume, by contradiction, that $C$ is such a cycle. We show how to construct an exchange of objects among the agents in $C$ such that their utility strictly increases without affecting the other agents. This will contradict the Pareto-optimality of $\mathbf{z}$.

Define $R := \pi(C)^{1/L}$; by assumption, $R < 1$.

For each $k\in[L]$, there is an edge from agent $i_k$ to object $o_k$. Hence, by the definition of \dcg{z}
\begin{itemize}
\item either $i_k$ consumes a positive amount of $o_k$ and both $i_k$ and $i_{k+1}$ agree that $o_k$ is a good ($v_{i_k,o_k}>0$ and $v_{i_{k+1},o_k}>0$),
\item or $i_{k+1}$ has  a positive amount of $o_k$ and both $i_k$ and $i_{k+1}$ agree that $o_k$ is a bad ($v_{i_k,o_k}<0$ and $v_{i_{k+1},o_k}<0$).
\end{itemize}
 Suppose each $i_k$  gives a small positive amount $ \varepsilon_k$ of $o_k$ to $i_{k+1}$ in the case of a good  or $i_{k+1}$
 gives $ \varepsilon_k$ fraction of $o_k$ to $i_k$ in the case of a bad ($ \varepsilon_k \in (0,h_k]$ where $h_k=z_{i_k,o_k}$ for a good and $h_k=z_{i_{k+1},o_k}$ for a bad).
 Then, agent $i_k$ loses a utility of $ \varepsilon_{k} \cdot |v_{i_k,o_k}|$, but gains  $ \varepsilon_{k-1} \cdot |v_{i_k, o_{k-1}}|$ from the previous agent, so the net change in the utility of $i_k$ is 
$ \varepsilon_{k-1} |v_{i_k, o_{k-1}}| 
-  \varepsilon_{k} |v_{i_k,o_k}|$ (where the arithmetic on the indices $k$ is done modulo $L$ in a way that the index is always in $\{1,\ldots,L\}$). To guarantee that all agents in $C$ strictly gain from the exchange, it is sufficient to choose $ \varepsilon_1,\ldots, \varepsilon_k$ such that the following inequalities hold for all $k\in[L]$:
\begin{equation}\label{ineq:k}
 \varepsilon_{k-1} |v_{i_k, o_{k-1}}| -  \varepsilon_{k} |v_{i_k,o_k}|>0 \iff
\frac{ \varepsilon_{k}}{ \varepsilon_{k-1}} <  \frac{|v_{i_k,o_{k-1}}|} {|v_{i_k, o_{k}}|}.
\end{equation}

For any $ \varepsilon_1 > 0$, define $ \varepsilon_k=  \varepsilon_{k-1} \cdot R\cdot \frac{|v_{i_k, o_{k-1}}|}{|v_{i_k,o_k}|}$ for $k\in\{2,\ldots,L\}$.
Since $R<1$, the inequality~\eqref{ineq:k} is satisfied for each $k\in\{2,\ldots,L\}$. It remains to show that it is satisfied for $k=1$, too (note that in this case $k-1=L$). Indeed:

$$ \varepsilon_L=  \varepsilon_1 \cdot R^{L-1} \cdot 
\prod_{k=2}^{L} \frac{|v_{i_k, o_{k-1}}|}{|v_{i_k,o_k}|}= \varepsilon_1 \cdot R^{L-1} \cdot 
\frac{|v_{i_1,o_1}|}{|v_{i_1, o_{L}}|}\prod_{k=1}^{L} \frac{|v_{i_k, o_{k-1}}|}{|v_{i_k,o_k}|}= \varepsilon_1\frac{R^{L-1}}{\pi(C)} \cdot 
\frac{|v_{i_1,o_1}|}{|v_{i_1, o_{L}}|}= \varepsilon_1 R^{-1}\frac{|v_{i_1,o_1}|}{|v_{i_1, o_{L}}|}.$$
Thus

$$\frac{ \varepsilon_1}{ \varepsilon_L} = R\frac{|v_{i_1, o_{L}}|}{|v_{i_1,o_1}|}<\frac{|v_{i_1, o_{L}}|}{|v_{i_1,o_1}|}.$$
By choosing $ \varepsilon_1$ sufficiently small, we guarantee $ \varepsilon_k \leq  h_k$ {for all $k\in[L]$}, so this trade is possible.
\end{proof}

\begin{proof}[Proof that no $C$ and no maliciousness $\implies$  existence of $\lambda$]
We assume that \dcg{z} contains no directed cycles $C$ with $\pi(C)<1$ and $\mathbf{z}$ is non-malicious. We prove the existence of weights $\lambda_i>0$ from Lemma~\ref{lem:enumerate-level2}. 

Add directed edges $i\to j$ between each pair of distinct agents $i,j\in[n]$. All the new edges have the same large positive weight in order to ensure that the new graph $\vec{G}$ has no cycles $C=(v_1\to v_2\to\dots\to v_{L+1}=v_1)$  with multiplicative weight $\pi(C)<1$. It is enough to pick

 $$w_{i\to j}=\left(\max\left\{1, \ |v_{k,o}|,\ \frac{1}{|v_{k,o}|} \ : k\in [n],\ o\in [m],\ v_{k,o}\ne 0    \right\}\right)^{2(n-1)}.$$
Indeed, any simple cycle $C$ containing a new edge has at most $2(n-1)$ old edges; if none of the old edges has weight zero, then $\pi(C)\geq 1$  by the definition of $w_{i\to j}$. 
Edges $k\to o$ with $w_{k\to o}=0$ cannot be a part of any cycle: by the definition of \dcg{z} such edges are possible only if $k$ consumes $o$ and $v_{k,o}=0$. Since $\mathbf{z}$ is non-malicious, such $o$  has no outgoing edges.


Fix an arbitrary agent, say, agent~$1$. For every other agent $j\in [n]$, let $P_{1,j}$ be a directed path from $1$ to $j$ in $\vec{G}$, for which the product $\pi(P_{1,j})$ is minimal. 
The minimum is well-defined and is attained on an acyclic path, since by the construction there are no cycles with a product smaller than $1$, so adding cycles to a path cannot make its product smaller.

Set the \emph{weight} of each agent $j$ as $\lambda_j := \pi(P_{1,j})$ (in particular $\lambda_1 = 1$). 
{We now show that these weights satisfy} the conditions~\eqref{eq:lm-weights-appendix}, namely: $z_{i,o}>0$ implies $\lambda_i v_{i,o}\geq\lambda_j v_{j,o}$ for all $j\in[n]$.
W.l.o.g., we can assume  that $i\ne j$ and both agents agree whether $o$ is a good or a bad, i.e., $v_{i,o}\cdot v_{j,o}>0$. Indeed, if agents disagree, then by the non-maliciousness, \eqref{eq:lm-weights-appendix} is satisfied with any  $\lambda_i,\lambda_j>0$.

In the case of a good ($v_{i,o}>0$ and $v_{j,o}>0$), there is an edge $i\to o$ (since $i$ consumes~$o$) and $o\to j$ (since $v_{j,o}>0$ and $z_{j,o}\ne 1$). Consider the optimal path $P_{1,i}$ and the concatenated path $Q_{1,j}= P_{1,i} \to o  \to j$.
The path $P_{1,j}$ has the minimal product among all paths from $1$ to $j$.  Therefore,

$$\pi(Q_{1,j}) \geq \pi(P_{1,j})\Longleftrightarrow \pi(P_{1,i})\cdot \frac{v_{i,o}}{v_{j,o}} \geq \pi(P_{1,j})\Longleftrightarrow \lambda_i v_{i,o}\geq \lambda_j v_{j,o}.$$
The mirror argument for a bad (both $v_{i,o}$ and $v_{j,o}$ are negative) is as follows. There is an edge $j\to o$ (because $v_{j,o}<0$ and $z_{j,o}\ne 1$) and $o\to i$ (since $i$ consumes $o$ and $o$ is a bad).  We define $Q_{1,i}$ as  $P_{1,j} \to o  \to i$ and get

$$\pi(Q_{1,i}) \geq \pi(P_{1,i})\Longleftrightarrow \pi(P_{1,j})\cdot \frac{|v_{j,o}|}{|v_{i,o}|} \geq \pi(P_{1,i})\Longleftrightarrow \lambda_j |v_{j,o}|\geq \lambda_i |v_{i,o}|\Longleftrightarrow \lambda_i v_{i,o}\geq \lambda_j v_{j,o}.$$
\end{proof}
\begin{proof}[Proof that existence of $\lambda$ $\implies$ fPO]
If in an allocation $\mathbf{z}$ each object $o$ is consumed by  agents $i$ with highest $\lambda_i v_{i,o}$, then  $\mathbf{z}$ itself maximizes the weighted sum of utilities $\sum_{i\in [n]} \lambda_i u_i(\mathbf{z}_i)$ over all allocations. Since  all $\lambda_i$ are positive, $\mathbf{z}$ is fPO because any Pareto-improvement must increase the weighted sum of utilities as well.
\end{proof}

Lemma \ref{lem:po-cycle} has a useful computational implication, Lemma \ref{lem:po-check}: 
\emph{
It is possible to decide in time {$O(nm(n+m))$} whether a given allocation $\mathbf{z}$ is fractionally Pareto-optimal.
}
\begin{proof}
The idea is the following: construct the graph $\dcg{z}$, replace each weight with its logarithm, 
and look for a negative cycle using one of many existing algorithms \citep{cherkassky1999negative} (e.g., Bellman-Ford). If there is a cycle $C$ in which the sum of log-weights is negative, then $\pi(C)<1$, so by Lemma \ref{lem:po-cycle}, $\mathbf{z}$ is not fractionally PO. Otherwise, $\mathbf{z}$ is fractionally PO. 
A negative cycle can be found in time $O(|V|\cdot|E|)$. Here $|V| = m+n$ and $|E|\leq mn$. 

Because of irrationality, logarithms can be computed only approximately and thus, to ensure the correctness of the algorithm, one has to adjust the quality of approximation depending on the input. However, these difficulties are easy to avoid by using a multiplicative version of 
any of the algorithms of \citet{cherkassky1999negative}:
multiplication replaces addition, division is used instead of subtraction, and one instead of zero.%
%

This allows one to avoid logarithms and keep the same bound of $O(nm(n+m))$ on runtime.
\end{proof}

\section{Worst-case Bound on Sharing.}
\label{sec:po-n-1}
In this section we prove Lemma \ref{lem:po-n-1}:
\emph{For any allocation $\mathbf{z}$, there exists
	a fractionally Pareto-optimal allocation $\mathbf{z^*}$ such that:
	\begin{itemize}
		\item (a) $\mathbf{z^*}$ weakly Pareto dominates $\mathbf{z}$, i.e., for any agent $i$, $u_i(\mathbf{z^*_i})\geq u_i(\mathbf{z_i})$.
		\item (b) the non-directed consumption graph $\ucg{z^*}$ is acyclic.
		\item (c) $\mathbf{z^*}$ has at most $n-1$ sharings (hence at most $n-1$ shared objects).
	\end{itemize}
	Such allocation $\mathbf{z^*}$ can be constructed in time $O(n^2 m^2 (n+m))$. }
\begin{proof}
If $\mathbf{z}$ is malicious, reallocate the objects: 
\begin{itemize}
	\item for each $o\in [m]$ with  $\max_{i\in [n]} v_{i,o}>0$, reallocate the shares of agents $j$ with $v_{j,o}\leq 0$ to an agent $i$ with $v_{i,o}>0$;
	\item for each $o\in [m]$ with $\max_{i\in [n]} v_{i,o}=0$, reallocate the shares of agents $j$ with $v_{j,o}<0$ to an agent $i$ with $v_{i,o}=0$.
\end{itemize}
Denote the resulting non-malicious allocation by $\mathbf{z}'$.

Let's call a cycle $C=(i_1 \to o_1 \to i_2 \to o_2 \to \allowbreak \ldots \to \allowbreak i_L \to o_L \to i_{L+1} = i_1)$ in the directed graph 
$\dcg{z'}$ 
\emph{simple} if each node is visited at most once and for any $i\in[n]$ and $o\in[m]$  only one of the edges $i\to o$ or $o\to i$ is contained in the cycle.

If there is a simple cycle $C$ in $\dcg{z'}$  with $\pi(C)\leq 1$, then $C$ can be eliminated by the cyclic trade making all the agents weakly better off (similarly to the  proof of Lemma~\ref{lem:po-cycle} in Appendix~\ref{sec:po-cycle}). Since both edges $i_k\to o_k$ and $o_k\to i_{k+1}$ exist in \dcg{z'}, the values $v_{i_{k},o_k}$ and $v_{i_{k+1},o_k}$ are both non-zero and have the same sign. We conduct the following transfers:
\begin{itemize}
	\item if $v_{i_{k},o_k}>0$ and $v_{i_{k+1},o_k}>0$ (i.e., $o_k$ is a good for $i_k$ and $i_{k+1}$), then take $\varepsilon_k$ amount of 
	$o_k$ from $i_k$ and give it to $i_{k+1}$ ($0<\varepsilon_k\leq h_k,$ where $ h_k =z_{i_k,o_k}$);
	\item if $v_{i_{k},o_k}<0$ and $v_{i_{k+1},o_k}<0$ (i.e., $o_k$ is a bad), then transfer $\varepsilon_k$ of $o_k$ from $i_{k+1}$ to $i_k$ ($0<\varepsilon_k\leq h_k=z_{i_{k+1},o_k}$).
\end{itemize}
The amounts $\varepsilon_k$ are selected in such a way that
 $\varepsilon_k |v_{i_k,o_k}|=\varepsilon_{k+1} |v_{i_k,a_{k+1}}|$ for $k\in[L-1]$. Hence, each agent $i_k$, $k=2,\dots,L$,  remains indifferent between the old and the new allocations while agent $i_1$ is weakly better off because of the condition $\pi(C)\leq 1$. We select epsilons as big as possible:
  $$\varepsilon_k=\frac{\prod_{q=1}^{k-1}  \frac{|v_{i_q,a_q}|}{|v_{i_q,a_{q+1}}|}}{\min_{l\in[L]} \left(\frac{1}{h_l} \prod_{q=1}^{l-1}  \frac{|v_{i_q,a_q}|}{|v_{i_q,a_{q+1}}|} \right)}, \qquad k\in [L],$$
  thus eliminating one of the edges $i_k\to o_k$ in \dcg{z'}:

Repeat this procedure again and again until there are no simple cycles with $\pi(C)\leq 1$. Note that we need at most $(n-1)m$ repetitions since each time at least one edge is deleted in the undirected graph $\ucg{z}$ and the total number of edges is at most $n\cdot m$. Denote the resulting allocation by $\mathbf{z^*}$. 

(a) By construction, $\mathbf{z^*}$ weakly improves the utility of each agent, is non-malicious, and has no cycles with $\pi(C)<1$.  Thus, $\mathbf{z^*}$ is fractionally Pareto-optimal by Lemma~\ref{lem:po-cycle}. 

(b) The undirected consumption graph of $\mathbf{z^*}$ is acyclic. Assume by contradiction that there is a cycle $C$ in $\ucg{z^*}$. Then in the directed graph $\dcg{z^*}$ there are two cycles: $C$ passed in one direction and in the opposite. Denote them by $\overleftarrow{C}$ and $\overrightarrow{C}$. Since $\pi(\overrightarrow{C})=\frac{1}{\pi(\overleftarrow{C})}$, by fractional-Pareto-optimality we get $\pi(\overrightarrow{C})=\pi(\overleftarrow{C})=1$; however all such cycles were eliminated in the previous stages of the algorithm.

(c) 
Since any acyclic graph on $m+n$ nodes has at most $m+n-1$ edges,  and the number of sharings equals the number of edges in \ucg{z^*} minus $m$,  the number of sharings at $\mathbf{z^*}$ is at most $n-1$.

It remains to estimate the complexity of the algorithm. 
Constructing the non-malicious allocation $\mathbf{z'}$ takes $O(n\cdot m)$ and the overall complexity is determined by the time needed to perform cyclic trades.
 Cycles with $\pi(C)<1$ can be found using the multiplicative modification of the Bellman-Ford algorithm, as in Lemma~\ref{lem:po-cycle} which results in $O(nm(n+m))$ operations per cycle. For a given cycle $C$ of length $L$ transfers are conducted in $O(L)=O(\min\{n,m\})$ since no simple cycle is longer than $2\min\{m,n\}$.
 
When all cycles $C$ with $\pi(C)<1$ have been eliminated, it remains to delete all the cycles in the \emph{undirected} consumption graph if any (note that all such cycles have $\pi(C)=1$). Such cycles can be found using a depth-first search which  needs $O(|V|+|E|)=O(n\cdot m)$ operations per cycle.

The total number of cycles to be eliminated is at most $(n-1)m$ and we get the upper bound $O(n^2 m^2 (n+m))$ for the overall time-complexity. 
\end{proof}


\ifdefined\OPRE
\end{APPENDICES}
\fi

\newpage

\bibliographystyle{informs2014} 

\end{document}